\documentclass[acmsmall,nonacm,screen]{acmart}

\usepackage{amsmath,amsthm,amsfonts,bbold,mathtools}
\usepackage{braket,xcolor}
\usepackage{tikz-cd}
\usepackage{stmaryrd}
\usepackage[nameinlink,capitalize]{cleveref}

\AtEndPreamble{%
  \theoremstyle{acmdefinition}
  \newtheorem{question}{Question}
}

\crefname{theorem}{Theorem}{Theorems}
\Crefname{theorem}{Theorem}{Theorems}
\crefname{lemma}{Lemma}{Lemmas}
\Crefname{lemma}{Lemma}{Lemmas}
\crefname{corollary}{Corollary}{Corollaries}
\Crefname{corollary}{Corollary}{Corollaries}
\crefname{proposition}{Proposition}{Propositions}
\Crefname{proposition}{Proposition}{Propositions}
\crefname{definition}{Definition}{Definitions}
\Crefname{definition}{Definition}{Definitions}
\crefname{question}{Question}{Questions}
\Crefname{question}{Question}{Questions}
\crefname{example}{Example}{Examples}
\Crefname{example}{Example}{Examples}
\crefname{appendix}{Appendix}{Appendices}
\Crefname{appendix}{Appendix}{Appendices}

\AddToHook{env/theorem/begin}{\crefalias{section}{theorem}}
\AddToHook{env/lemma/begin}{\crefalias{theorem}{lemma}}
\AddToHook{env/corollary/begin}{\crefalias{theorem}{corollary}}
\AddToHook{env/proposition/begin}{\crefalias{theorem}{proposition}}
\AddToHook{env/definition/begin}{\crefalias{theorem}{definition}}
\AddToHook{env/example/begin}{\crefalias{theorem}{example}}

\newcommand{\midv}{\,\middle\vert\,}

\newcommand{\Z}{\mathbb Z}
\newcommand{\Co}{\mathbb C}

\newcommand{\bbone}{\mathbb 1}
\newcommand{\Id}{\bbone}

\newcommand{\tr}{\operatorname{tr}}
\newcommand{\Span}{\operatorname{span}}

\DeclarePairedDelimiter\parens{\lparen}{\rparen}
\DeclarePairedDelimiter\abs{\lvert}{\rvert}
\DeclarePairedDelimiter\norm{\lVert}{\rVert}
\DeclarePairedDelimiter\floor{\lfloor}{\rfloor}

\DeclarePairedDelimiter\braces{\lbrace}{\rbrace}

\newcommand{\calA}{\mathcal{A}}
\newcommand{\calB}{\mathcal{B}}
\newcommand{\calC}{\mathcal{C}}
\newcommand{\calD}{\mathcal{D}}
\newcommand{\calH}{\mathcal{H}}
\newcommand{\calK}{\mathcal{K}}

\newcommand{\calO}{\mathcal{O}}
\newcommand{\calU}{\mathcal{U}}
\newcommand{\calV}{\mathcal{V}}

\setcopyright{none}

\begin{document}
 
\title{Quantisation of Abstract Data Types} 

\author{Mingsheng Ying}
 \affiliation{
  \department{Centre for Quantum Software and Information}          
  \institution{University of Technology Sydney}          
    \city{Sydney}
  \country{Australia}
}
\email{Mingsheng.Ying@uts.edu.au}       

\author{Zhicheng Zhang}
 \affiliation{
  \department{School of Information and Communication Technology}          
  \institution{Griffith University}          
    \city{Brisbane}
  \country{Australia}
}
\email{iszczhang@gmail.com}

\author{Kean Chen}
 \affiliation{
  \department{Department of Computer and Information Science}          
  \institution{University of Pennsylvania}          
    \city{Philadelphia}
  \country{USA}
}
\email{keanchen.gan@gmail.com}

\begin{abstract} 
    In this paper, we introduce a notion of abstract quantum data type within the framework of universal algebra. This notion provides an algebraic foundation for describing data abstraction in quantum programming. We formally define a quantisation of classical data types and show that their equational specifications can be soundly lifted to the quantum setting.
    Two standard quantisation methods for classical functions, namely the bit oracle and the phase oracle, arise as special cases of this general construction.
    We illustrate the framework with applications to quantum arrays and quantum error-correcting codes, showing how they can be understood through the lens of data-type quantisation.
    We further establish conditions under which quantisation preserves structural relationships and constructions of classical data types, including embeddings, isomorphisms, and products.
\end{abstract}



\maketitle
 
\section{Introduction}

\begin{align*}\qquad\qquad\qquad\qquad\qquad\qquad &\mathit{Algorithms\ +\ Data\ Structures = Programs}\\  &\qquad\qquad \qquad\qquad \qquad\qquad  \mbox{--- Niklaus Wirth \cite{Wirth}}
\end{align*}

Quantum programming research has been conducted for more than 30 years, beginning with Knill’s proposal of a set of conventions for writing quantum pseudocode \cite{Knill}, and it has become increasingly active in recent years \cite{Heim, Ying24}. The majority of this research has focused on the algorithmic aspect of quantum programs. However, programming methodology must also study the structures, representations, and operations of data involved in large and complex programs \cite{Wirth}. To date, the data aspect of quantum programming has been less investigated. 

\subsection{Data types and data structures in quantum programming}

Several basic quantum data types that abstract qubits and quantum registers are widely used in practical quantum programming languages, e.g. Qiskit \cite{qiskit}, Q\# \cite{sharp}, Cirq \cite{Cirq}, Quipper \cite{Quipper}, Guppy \cite{Guppy}, Silq \cite{Silq}, Qrisp \cite{Qrisp} and isQ \cite{isQ}. Some more advanced quantum data types have been defined in theoretical research on quantum programming; for example, inductive quantum data types were introduced by P\'{e}choux, Perdrix, Rennela and Zamdzhiev \cite{Pech}. 
On the other hand, from the viewpoint of developer-friendly programming, a quantum programming language, Rhyme, with richer data types—including quantum extensions of floats, characters, arrays, and strings—was proposed by Varga, Aragon\'{e}s-Soria and Oriol \cite{Var}. 

Various quantum data structures have been introduced for the design of quantum algorithms (see e.g. \cite{Amb03, DS1, DS2, DS3, DS4}). Yuan and Carbin's Tower \cite{Yuan} provides the first programming language support for data structures in quantum superposition, including using random-access memory and recursive definitions of operations. A library of quantisations of data structures \textit{lists}, \textit{stacks}, \textit{queues}, \textit{strings}, and \textit{sets} is programmed in Tower.

\subsection{Data abstraction: from classical to quantum programming} 

Data abstraction is a foundational concept in classical programming that allows a programmer to focus on what an object does rather than how it does it. It is the process of hiding complex implementation details while exposing only the essential features of an object or data type. Thus, it allows developers to work with abstract data types (e.g. integers, arrays, or linked lists) without worrying about memory layout or hardware details \cite{Liskov, Liskov87}. 

We believe that data abstraction will play an even more critical role in quantum programming \cite{AbsArt, AbsArt1} for at least the following two reasons:\begin{enumerate}\item Concepts such as superposition, entanglement, and measurement impose physical constraints that quantum data must respect but programmers may overlook or ignore.
Therefore, data abstraction is essential for writing correct and scalable quantum software.
\item Hardware independence is essential in quantum programming because quantum hardware is evolving rapidly. High-level abstractions allow quantum algorithms to be expressed without being tied to a specific quantum device or gate set. Compilers then map these abstract representations onto real hardware, optimising for noise, connectivity, and error rates.\end{enumerate}

Building upon previous work on quantum data structures \cite{Yuan} and quantum data types \cite{Var}, we then naturally ask the following:

\begin{question}\label{Q1} How can we realise data abstraction in quantum programming?\end{question}

\subsection{Abstract data types} 

The introduction of abstract data types (ADTs) is a key way of realising data abstraction \cite{Gut}. An ADT specifies a data type by focusing on what it does rather than how it is implemented, and thus the internal representation is hidden from the user. 
ADTs are defined from a user's viewpoint. In contrast, data structures are defined from an implementer's viewpoint as concrete representations of data. ADTs encourage programmers to think at a higher level of abstraction. The main advantages of ADTs include:
\begin{itemize}\item By separating the interface from the implementation, programmers can change how an ADT is implemented without affecting the rest of the program. For example, a stack ADT is defined by operations such as push, pop, and peek, regardless of whether it is implemented using an array or a linked list. This makes code easier to maintain, understand, and reuse. 
\item ADTs support modularity so that users interact only through well-defined operations. This can significantly improve software reliability because errors caused by directly manipulating data are reduced. 
\end{itemize}

Turning to the realm of quantum programming, \Cref{Q1} should be concretised to the following: 
  
\begin{question}\label{Q2} How can we define \textit{abstract quantum data types}? Is there a principled way to lift classical data types to the quantum setting---what one might call the \textit{quantisation of data types}?\end{question}

\subsection{Contributions of this paper}

The aim of this paper is to introduce a method for the quantisation of data types that supports data abstraction in quantum programming. The main difference between our work and previous studies \cite{Yuan, Var} is that our approach operates at the level of abstract data types, whereas \cite{Yuan, Var} focus on less abstract levels, such as data structures and concrete data types. The technical contributions of this paper are threefold:
\begin{enumerate}
\item We formally define the notion of an abstract quantum data type within a universal-algebraic framework.
\item We introduce a quantisation of classical ADTs and prove that their equational specifications can be soundly lifted to the quantum setting. This notion includes generalisations of two basic quantisation methods for classical functions as important special cases.
\item We identify conditions ensuring that quantisation is compatible with structural relationships and constructions of classical data types, including embeddings, isomorphisms, and products.
\end{enumerate}
We complement these results with illustrative applications to quantum arrays and quantum error correction.

We hope that the introduction of abstract quantum data types will help programmers design clean, flexible, and robust software systems for quantum computers. In particular, \textit{inheritance} is a fundamental feature of object-oriented programming \cite{Cardelli, Goguen}, playing a crucial role in managing software complexity through mechanisms such as code reuse and hierarchical classification. The embedding result in Contribution (3) shows that certain algebraic inheritance relationships between classical data types can, under suitable conditions, be preserved under quantisation and transferred to their quantum counterparts.

\subsection{Organisation of the paper}

For the convenience of the reader, the basic concepts of ADTs are reviewed in \Cref{sec-CDT}. The notion of a quantum data type is introduced in \Cref{sec-QDT}, where several simple examples are also given to illustrate how quantum ADTs can be used in quantum programming. 
\Cref{sec-quantisation} develops a quantisation of classical ADTs and shows how equational specifications can be lifted to the quantum setting.
It also presents two standard quantisation methods as special cases and explains their relationship through phase kickback.
The framework is then illustrated through quantum arrays in \Cref{sec-app} and quantum error correction in \Cref{sec-qec}.
\Cref{sec-emb,sec-prod} investigate the structural properties of quantisation through embeddings and products, respectively. The paper concludes with a summary of the results and a brief discussion about topics for further research. The proofs of the main theorems and propositions are deferred to \Cref{appendix}.

\section{ADTs in Classical Computing}\label{sec-CDT}

In this section, we briefly review the basics of ADTs, following the textbook \cite{EM}. For more details, the reader can consult related chapters of \cite{EM}. 

A data type is a collection of data domains, designated basic data items, and operations on these domains. However, a programmer may not be interested in the concrete representation of a data type but only in its properties at an abstract level.  
Thus, the notion of ADT is introduced to specify the kinds of data that can be stored, the operations that can be performed on that data, and the rules those operations satisfy, independently of their concrete representation.

\subsection{Basics of Universal Algebra}

ADTs can be properly formalised in the mathematical language of universal algebra. A signature specifies the interface of an ADT. Formally, we have: 
\begin{definition}[Signatures]\label{def-sig} A signature is a pair $\mathit{Sig}=(S,\mathit{OP})$, where: \begin{enumerate}\item $S$ is a set of sorts;
\item $\mathit{OP}$ is a set of operation symbols. Each symbol $O\in\mathit{OP}$ is assigned input sorts $\mathit{in}(O)\in S^\ast$ (the set of strings of elements in $S$, including the empty string $\epsilon$) and an output sort $\mathit{out}(O)\in S$. 
\end{enumerate}
\end{definition}

We often write $O:s_1,\ldots,s_n\rightarrow s$ for $\mathit{in}(O)=s_1,\ldots,s_n$ and $\mathit{out}(O)=s$. In particular, a constant symbol is an operation symbol $c:\rightarrow s$ where the input sort $\mathit{in}(c)=\epsilon$ is empty. 

Mathematically, a data type can then be modelled as an algebra as defined as follows: 

\begin{definition}[Algebras]\label{def-alg} 
    An algebra $A$ of a signature $\mathit{Sig}=(S,\mathit{OP})$, or a $\mathit{Sig}$-algebra for short, consists of:
    \begin{enumerate}
        \item for each $s\in S$, $A_s$ is a nonempty set, called the domain of sort $s$; 
        \item for each $O\in\mathit{OP}$ with $O:s_1,\ldots,s_n\rightarrow s$, $O^A$ is a mapping $A_{s_1}\times\cdots\times A_{s_n}\rightarrow A_s$, called an operation. In particular, for each constant symbol $c:\rightarrow s$, $c^A$ is an element of $A_s$.  
    \end{enumerate}
\end{definition}

Throughout this paper, we only consider finite domains.
The rules that the operations of an ADT must follow are usually specified as a family of equations. To define them, for each $s\in S$, we assume a set $X_s$ of variables of sort $s$. The variable sets $X_s$ are required to be pairwise disjoint and also disjoint from the operation symbols $\mathit{OP}$. Then we can introduce the following definition.

\begin{definition}[Terms]\label{def-term} The terms are recursively defined as follows:\begin{enumerate}\item every constant symbol $c:\rightarrow s$ and every variable $x\in X_s$ are terms of sort $s$;
\item if $t_1,\ldots,t_n$ are terms of sorts $s_1,\ldots,s_n$, respectively, and $O:s_1,\ldots,s_n\rightarrow s\in\mathit{OP}$, then $O(t_1,\ldots,t_n)$ is a term of sort $s$.
\end{enumerate}
\end{definition}

A term without variables is called a ground term. The semantic interpretation of terms in a given algebra is presented through the notion of evaluation defined in the following: 

\begin{definition}[Evaluation] Let $A$ be an $\mathit{Sig}$-algebra and $\sigma$ an assignment in $A$ that maps each variable $x\in X_s$ to an element $\sigma(x)\in A_s$ for every sort $s$. Then the evaluation $\mathit{eval}_\sigma(t)$ of a term $t$ in assignment $\sigma$ is recursively defined as follows: \begin{enumerate}\item if $t$ is a variable $x$, then $\mathit{eval}_\sigma(t)=\sigma(x)$, and if $t$ is a constant $c$, then $\mathit{eval}_\sigma(t)=c^A$;
\item if $t=O(t_1,\ldots,t_n)$, then $\mathit{eval}_\sigma(t)=O^A(\mathit{eval}_\sigma(t_1),\ldots,\mathit{eval}_\sigma(t_n))$.
\end{enumerate}
\end{definition}

As noted above, a large class of the rules for the operations in an ADT can be described by equations between terms. The notion of equations can be formally defined in the following: 

\begin{definition}[Equations and Validity] 
\quad
\begin{enumerate}\item Let $t_1$ and $t_2$ be two terms of the same sort $s$. If all variables occurring in $t_1$ and $t_2$ are among $x_1,\ldots,x_n$, then \begin{equation}\label{eq-0}e= (\forall x_1:s_1, \ldots, x_n:s_n) (t_1=t_2)\end{equation} is called an equation, where $x_i$ is a variable of sort $s_i$ for every $i=1,\ldots,n$.  
\item The equation $e$ is valid in an algebra $A$, or $A$ satisfies $e$, written $A\models e$, if for any assignment $\sigma$ in $A$, we have $\mathit{eval}_\sigma(t_1)=\mathit{eval}_\sigma(t_2)$. \end{enumerate}
\end{definition}

It is easy to see that the validity of equation $e$ only depends on the values $\sigma(x_i)$ $(i=1,\ldots,n)$ of assignments $\sigma$. In the specifications of ADTs in practical applications, universal quantifiers $(\forall x_1:s_1,\ldots,x_n:s_n)$ in \Cref{eq-0} are often omitted (see \Cref{exam-adt}).  

If both $t_1$ and $t_2$ are ground terms, then $e$ is called a ground equation. The validity of a ground equation $e$ in an algebra does not depend on any assignment $\sigma$ in $A$. 

\subsection{Specification of ADTs}

Now we can use the algebraic framework defined in the previous subsection to specify ADTs. 

\begin{definition}[Specifications]\label{def-spec}
    \quad
    \begin{enumerate}\item A specification is a triple $\mathit{Spec}=(S,\mathit{OP},E)$, where $\mathit{Sig}=(S,\mathit{OP})$ is a signature, and $E$ is a set of equations over the signature $\mathit{Sig}$.
    \item An algebra of the specification $\mathit{Spec}$, or $\mathit{Spec}$-algebra for short, is an algebra $A$ of signature $\mathit{Sig}$ such that $A\models e$ for every $e\in E$.\end{enumerate}
\end{definition}

Besides the equational specifications defined above, there are more sophisticated classes of ADT specifications, such as Horn specifications.
In this paper, we mainly consider equational specifications. 

For a better understanding of the notions introduced above, let us see the following simple example: 

\begin{example}[Abstract data type of strings]\label{exam-adt} A string is a sequence of basic data items from a given domain, called the alphabet. Let $c_1,\ldots,c_n$ be constant symbols denoting the basic data items. The most useful operations of strings in practical applications are:\begin{itemize}\item $\mathit{make}$: it converts each basic data item into a string of length $1$;
\item $\mathit{concat}$: it concatenates two strings;
\item $\mathit{ladd}$, $\mathit{radd}$: they add a basic data item to the left and right ends of a string, respectively.  
\end{itemize} 
A specification of strings is given as follows:\begin{align*}\texttt{STRING} =\ &\textbf{sorts}: \texttt{alphabet}, \texttt{string}\\
&\textbf{operations}: c_1,\ldots,c_n:\rightarrow \texttt{alphabet}\\ &\qquad\ \ \ \mathit{empty}:\rightarrow\texttt{string}\\ 
&\qquad\ \ \ \mathit{make}:\texttt{alphabet}\rightarrow\texttt{string}\\ &\qquad\ \ \ \mathit{concat}:\texttt{string}, \texttt{string}\rightarrow\texttt{string}\\ &\qquad\ \ \ \mathit{ladd}: \texttt{alphabet},\texttt{string}\rightarrow\texttt{string}\\ 
&\qquad\ \ \ \mathit{radd}: \texttt{string},\texttt{alphabet}\rightarrow\texttt{string}\\
&\textbf{equations}: a:\texttt{alphabet}, s,s_1,s_2,s_3:\texttt{string}\\ &\qquad\ \ \ \mathit{concat}(s,\mathit{empty})=s\\ &\qquad\ \ \ \mathit{concat}(\mathit{empty},s)=s\\ &\qquad\ \ \ \mathit{concat}(\mathit{concat}(s_1,s_2),s_3)=\mathit{concat}(s_1,\mathit{concat}(s_2,s_3))\\ &\qquad\ \ \ \mathit{ladd}(a,s)=\mathit{concat}(\mathit{make}(a),s)\\ &\qquad \ \ \ \mathit{radd}(s,a)=\mathit{concat}(s,\mathit{make}(a))
\end{align*}

We can define a $\texttt{STRING}$-algebra $A$ as follows:
\begin{itemize}
    \item
        Let $A_{\texttt{alphabet}}=\braces*{a_1,\ldots,a_n}$ be a domain for the alphabet.
        Let $A_{\texttt{string}}=(A_{\texttt{alphabet}})^*$ be the set of strings of elements in $A_{\texttt{alphabet}}$, including the empty string $\epsilon$.
    \item
        For constant symbols, define $c_i^A=a_i$ for $i=1,\ldots,n$ and $\mathit{empty}^A=\epsilon$.
    \item
        For nonconstant operation symbols, define
        \begin{enumerate}
            \item 
            $\mathit{make}^A:A_{\texttt{alphabet}}\rightarrow A_{\texttt{string}}$, $\mathit{make}^A(a)=a$ for all $a\in A_{\texttt{alphabet}}$; 
            \item $\mathit{concat}^A:A_{\texttt{string}}\times A_{\texttt{string}}\rightarrow A_{\texttt{string}}$, $\mathit{concat}^A(u,v)=uv$ for all $u,v\in A_{\texttt{string}}$; 
        \item $\mathit{ladd}^A:A_{\texttt{alphabet}}\times A_{\texttt{string}}\rightarrow A_{\texttt{string}}$, $\mathit{ladd}^A(a,u)=au$ for all $a\in A_{\texttt{alphabet}}$ and $u\in A_{\texttt{string}}$; \item $\mathit{radd}^A:A_{\texttt{string}}\times A_{\texttt{alphabet}}\rightarrow A_{\texttt{string}}$, $\mathit{radd}^A(u,a)=ua$ for all $a\in A_{\texttt{alphabet}}$ and $u\in A_{\texttt{string}}$. 
        \end{enumerate} 
\end{itemize}
It is easy to check that the above operations satisfy the equations in $\texttt{STRING}$, and hence $A$ is a $\texttt{STRING}$-algebra.
\end{example}

\section{Quantum Data Types}\label{sec-QDT}

From now on, we generalise the algebraic specification techniques for classical ADTs into quantum computing. In this section, we formally define quantum data types in the universal-algebraic framework, analogously to their classical counterparts in \Cref{sec-CDT}, and illustrate it with some simple examples. 

\subsection{Basic Definitions}

Let us first introduce the notion of a quantum signature as the quantum counterpart of \Cref{def-sig}: 

\begin{definition}[Quantum Signatures] A quantum signature is a triple $\mathit{QSig}=(S_q,\calK, \calU)$, where: \begin{enumerate}\item $S_q$ is a set of quantum sorts;
\item $\calK$ is a set of quantum state symbols. Each symbol $\ket{\psi} \in \calK$ is assigned sort $\mathit{sort}(\ket{\psi})\in S_q^+$ (the set of nonempty strings of elements in $S_q$);
\item $\calU$ is a set of quantum gate symbols. Each symbol $U \in \calU$ is assigned sort $\mathit{sort}(U)\in S_q^+$;
\end{enumerate}\end{definition}

Roughly speaking, $\calK$ and $\calU$ correspond to the constant symbols and operation symbols in \Cref{def-sig}, respectively. But there are some subtle differences between them and their classical counterparts. For a classical constant symbol $c:\rightarrow s$, its input sort is empty and its output sort is a single sort $s\in S$. However, for a quantum state symbol $\ket{\psi}$, it is possible that $\mathit{sort}(\ket{\psi})=s_1,\ldots,s_n\in S_q^+$ with $n>1$. The reason is that $\ket{\psi}$ may denote an entangled state between $n$ subsystems, which cannot be decomposed into a string of the states of individual subsystems. We often write $\ket{\psi}: s_1,\ldots,s_n$ for $\mathit{sort}(\ket{\psi})=s_1,\ldots,s_n$.
Moreover, for a classical operation symbol $O:s_1,\ldots,s_n\rightarrow s$, the input sorts $s_1,\ldots,s_n$ and the output sort $s$ need not coincide. In contrast, the input and output sorts of a quantum gate symbol $U$ must be the same because it denotes a unitary operator. We write $U:s_1,\ldots,s_n$ for $\mathit{sort}(U)=s_1,\ldots,s_n$. Indeed, $\mathit{sort}(U)$ is both the input sort and output sort of $U$.

As the quantum counterparts of  algebras described in \Cref{def-alg}, quantum algebras are then defined by taking the domains of sorts to be Hilbert spaces and interpreting quantum state symbols and quantum gate symbols as pure quantum states and unitary operators, respectively. 
 
\begin{definition}[Quantum Algebras] A quantum algebra $\calA$ of signature $\mathit{QSig}=(S_q,\calK,\calU)$ consists of: \begin{enumerate}\item for each $s\in S_q$, $\calH^\calA_s$ is a Hilbert space, called the quantum domain of sort $s$ (the superscript $\calA$ of $\calH^\calA_s$ is often omitted for simplicity); 
\item for each quantum state symbol $\ket{\psi}\in\calK$ with $\mathit{sort}(\ket{\psi})=s_1,\ldots,s_n$, $\ket{\psi}^\calA$ is a unit vector (i.e. a pure quantum state) in $\calH_{s_1}\otimes\cdots\otimes\calH_{s_n}$; and
\item for each quantum gate symbol $U\in\calU$ with $\mathit{sort}(U)=s_1,\ldots,s_n$, $U^\calA$ is a unitary operator (i.e. a quantum gate) on $\calH_{s_1}\otimes\cdots\otimes\calH_{s_n}$.\end{enumerate}
\end{definition}

In this paper, we only consider quantum algebras where state symbols are interpreted as pure states and gate symbols as unitary operators. Of course, Definition 3.2 can be generalized such that state symbols are interpreted as mixed states (i.e. density operators) and gate symbols as noisy quantum gates (i.e. quantum channels) when needed in applications. 

Given a quantum signature $\mathit{QSig}$, we assume a set $Q$ of quantum variables. Each $q\in Q$ is assigned a sort $\mathit{sort}(q)\in S_q$. For a string $\overline{q}=q_1,\ldots,q_n$ of quantum variables,
we define $\mathit{sort}(\overline{q})=\overline{s}=s_1,\ldots,s_n$, where $s_i=\mathit{sort}(q_i)$.
When the context is clear, we also use $\overline{q}$ to denote the set $\braces{q_1,\ldots,q_n}$.
Then the notion of term is generalised to the quantum case as follows: 
\begin{definition}[Quantum Circuit Formulas and Terms]  
    \quad
    \begin{enumerate}
        \item Quantum circuit formulas $C$ and their quantum variables $\mathit{qv}(C)$ are inductively defined as follows:\begin{enumerate}
\item if $U:s_1,\ldots,s_n$ is a quantum gate symbol and $\overline{q}=q_1,\ldots,q_n$ is a string of distinct quantum variables with $\mathit{sort}(q_i)=s_i$ for $i=1,\ldots,n$, then $C=U[\overline{q}]$ is a quantum circuit formula, and $\mathit{qv}(C)=\overline{q}$; \item if $C_1$ and $C_2$ are quantum circuit formulas, then $C =C_1;C_2$ is a quantum circuit formula, and $\mathit{qv}(C)=\mathit{qv}(C_1)\cup \mathit{qv}(C_2)$.
\end{enumerate}
\item Suppose $C$ is a quantum circuit formula. For $i=1,\ldots,n$, let $\ket{\psi_i}:\overline{s_i}$ be a quantum state symbol and $\overline{q_i}$ be a string of distinct quantum variables with $\mathit{sort}(\overline{q_i})=\overline{s_i}$ and $\overline{q_i}\cap \overline{q_j}=\emptyset$ for $i\neq j$.
Then $\tau=C \ket{\psi_1}_{\overline{q_1}}\ldots\ket{\psi_n}_{\overline{q_n}}$ is a quantum term, and $\mathit{qv}(\tau)=\overline{q_1}\cup \ldots \cup\overline{q_n}\cup\mathit{qv}(C)$.
When the context is clear, $\ket{\psi}_{\overline{q}}$ may be used to denote the string $\ket{\psi_1}_{\overline{q_1}}\ldots\ket{\psi_n}_{\overline{q_n}}$ of quantum state symbols, where $\overline{q}=\overline{q_1},\ldots,\overline{q_n}$. Define the set of free variables in quantum term $\tau$ to be $\mathit{fv}(\tau)=\mathit{qv}(\tau)\setminus \overline{q}$.
\end{enumerate}
\end{definition}

The definition of quantum terms above differs slightly from its classical counterpart in \Cref{def-term}. Obviously, a quantum circuit formula $C$ stands for a (combinational) quantum circuit constructed from quantum gates denoted by elements of $\calU$ applied to quantum variables in $Q$. Furthermore, a quantum term $C\ket{\psi_{1}}_{\overline{q_1}}\ldots \ket{\psi_n}_{\overline{q_n}}$ stands for a quantum circuit formula $C$ with quantum variables $\overline{q}=\overline{q_1},\ldots,\overline{q_n}$ initialised in the state $\ket{\psi_1}\ldots\ket{\psi_n}$, where each $\ket{\psi_i}$ can be entangled. Conversely, a quantum circuit formula $C$ can be seen as a quantum term with $\overline{q}=\emptyset$. A string of quantum state symbols $\ket{\psi}_{\overline{q}}=\ket{\psi_{1}}_{\overline{q_1}}\ldots\ket{\psi_n}_{\overline{q_n}}$ can also be seen as a quantum term with empty $C$.

Given an algebra $\calA$ of a quantum signature $\mathit{QSig}$, for any subset $Q^\prime\subseteq Q$ of quantum variables, we write: $$\calH_{Q^\prime}=\bigotimes_{q\in Q^\prime}\calH_{\mathit{sort}(q)}$$ for the Hilbert space of the composite system $Q^\prime$. Then an assignment of quantum variables in $Q^\prime$ can be understood as a quantum state $\ket{\varphi}\in\calH_{Q^\prime}$, although there is a basic difference between it and an assignment of classical variables; that is, the values of two classical variables are independent, but $\ket{\varphi}$ could contain entanglement. Furthermore, the notion of evaluation can be generalised straightforwardly to the quantum case.  
As usual, for a unitary operator $U$, we implicitly identify it with its extension $U\otimes \Id$ by tensoring with an identity operator $\Id$.

\begin{definition}[Evaluation]  
    \quad
    \begin{enumerate}\item Let $C$ be a quantum circuit formula. Then its evaluation $\mathit{eval}(C)$ in algebra $\calA$ is a unitary operator on $\calH_{\mathit{qv}(C)}$ recursively defined as follows:\begin{enumerate}\item if $C=U[\overline{q}]$ for some $U\in \calU$, then $\mathit{eval}(C)=U^\calA_{\overline{q}}$ (unitary operator $U^\mathcal{A}$ acting on the quantum system denoted by $\overline{q}$); \item if $C=C_1;C_2$, then $\mathit{eval}(C_1;C_2)=\mathit{eval}(C_2)\circ \mathit{eval}(C_1)$ is a composition of unitary operators. 
\end{enumerate}\item Let $\tau=C\ket{\psi_1}_{\overline{q_1}}\ldots \ket{\psi_n}_{\overline{q_n}}$ be a quantum term. Suppose $\overline{q}=\overline{q_1},\ldots,\overline{q_n}$ and $\overline{r}$ is a string of quantum variables such that $\overline{r}\supseteq \mathit{fv}(\tau)$ and $\overline{r}\cap \overline{q}=\emptyset$. Let $\ket{\varphi}$ be a quantum state in $\calH_{\overline{r}}$. Then the evaluation of $\tau$ in state $\ket{\varphi}$ is defined as the quantum state: 
\begin{equation}\label{eval-state}\mathit{eval}_{\ket{\varphi}}(\tau)=\mathit{eval}(C)\parens*{\ket{\psi_1}_{\overline{q_1}}^{\calA}\otimes \ldots\otimes \ket{\psi_n}_{\overline{q_n}}^\calA\otimes\ket{\varphi}_{\overline{r}}}.
\end{equation} 
Here, $\mathit{eval}(C)=\Id$ if $C$ is empty.
\end{enumerate}\end{definition}

The notion of equation can be smoothly generalised to the quantum setting too: 

\begin{definition}[Equations]\label{def-qeq} Let $\tau_1=C_1\ket{\psi_1}_{\overline{q_1}}$ and $\tau_2=C_2\ket{\psi_2}_{\overline{q_2}}$ be two quantum terms that satisfy $\overline{q_1}\cap \mathit{fv}(\tau_2)=\overline{q_2}\cap \mathit{fv}(\tau_1)=\emptyset$. Suppose $\overline{r}=\mathit{fv}(\tau_1)\cup \mathit{fv}(\tau_2)$ and $\mathit{sort}(\overline{r})=\overline{s}$. Then, for $\overline{p}\subseteq\mathit{qv}(\tau_1)\cap \mathit{qv}(\tau_2)$,
\begin{equation}\label{equation}e=(\forall \overline{r}:\overline{s})(\tau_1 = \tau_2)\Downarrow\overline{p}\end{equation} is called an equation, and $\overline{p}$ are called the principal variables of $e$. In particular, if $\overline{p}=\mathit{qv}(\tau_1)\cap\mathit{qv}(\tau_2)$, then $\Downarrow\overline{p}$ can be dropped from \Cref{equation}.\end{definition}

Similar to the case of classical equations, universal quantifiers $(\forall \overline{r}:\overline{s})$ in \Cref{equation} are often omitted in practical specifications of quantum ADTs (see \Cref{ex-qbit} and \Cref{prop-array}).

It is worth examining carefully the differences between classical \Cref{eq-0} and quantum \Cref{equation}. First, we notice that in terms $\tau_1$ and $\tau_2$, quantum variables $\overline{q_1}$ and $\overline{q_2}$ are already initialised, but $\overline{r}$ are not. Thus, \Cref{equation} means that the equality between $\tau_1$ and $\tau_2$ is valid for all possible input states of $\overline{r}$. Second, an input state $\ket{\varphi}_{\overline{r}}$ of $\overline{r}=r_1,\ldots,r_n$ may be entangled. 
Third, the evaluation of $\tau_1$ (resp.\ $\tau_2$) gives an output state $\ket{\psi_1}$ (resp.\ $\ket{\psi_2}$) of multiple quantum variables, but one may only be interested in a subset $\overline{p}$ of them. Thus, the variables in $\overline{p}$ are designated as the principal variables. Consequently, to define the validity of equations in a quantum algebra, we need the notion of a reduced density operator. For a quantum state $\ket{\psi}$ in Hilbert space $\calH_{\overline{p}}\otimes\calH_{\overline{p^\prime}}$, the restriction of $\ket{\psi}$ onto $\overline{p}$ is defined as the reduced density operator 
$$\ket{\psi}\downarrow\overline{p}=\tr_{\overline{p^\prime}}(\ket{\psi}\!\bra{\psi})$$ where $\tr_{\overline{p^\prime}}$ stands for the partial trace over subsystem $\overline{p^\prime}$ (see Section 2.4.3 of \cite{NC00} for its precise definition). Intuitively, $\ket{\psi}\downarrow\overline{p}$ describes the state of subsystem $\overline{p}$ when the whole system $\overline{p}\uplus\overline{p^\prime}$ is in state $\ket{\psi}.$ With this preparation, we are able to introduce: 
 
\begin{definition}[Validity]
    \label{def:validity}
    The equation $e$ given in \Cref{equation} is valid in algebra $\calA$, written $\calA\models e$, if for any quantum state $\ket{\varphi}\in\calH_{\overline{r}}$, we have:
    \begin{equation*}
        \mathit{eval}_{\ket{\varphi}}(\tau_1)\downarrow\overline{p}=\mathit{eval}_{\ket{\varphi}}(\tau_2)\downarrow\overline{p}.
    \end{equation*}
\end{definition}

\subsection{Specification of Quantum Data Types}

As one may expect, the algebraic framework developed in the previous subsection can be used to specify ADTs in quantum computing. The following definition is a straightforward generalisation of \Cref{def-spec}. 

\begin{definition}[Quantum Specification]
    \quad
    \begin{enumerate}\item A quantum specification is a quadruple $\mathit{QSpec}=(S_q,\calK,\calU,E)$, where $\mathit{QSig}=(S_q,\calK,\calU)$ is a quantum signature, and $E$ is a set of equations of the form in \Cref{equation}. 
\item A $\mathit{QSpec}$-algebra is a quantum algebra $\calA$ of signature $\mathit{QSig}$ such that $\calA\models e$ for every $e\in E$. 
\end{enumerate}
\end{definition}

For a better understanding of quantum ADTs, let us first see the following simple example: 

\begin{example}[Abstract data type of qubits]\label{ex-qbit}
Consider the following quantum specification:
\begin{align*}\texttt{QBIT} =\  
&\textbf{sorts}:\ \texttt{qbit} \\ 
&\textbf{states}:\ \ket{0}, \ket{1}, \ket{+}, \ket{-}: \texttt{qbit} \\ 
&\qquad\quad \ket{\beta_{00}},\ket{\beta_{01}},\ket{\beta_{10}},\ket{\beta_{11}}: \texttt{qbit}, \texttt{qbit}\\ 
&\textbf{gates}:\ \mathit{NOT},H:\texttt{qbit}\\ 
&\qquad\quad \mathit{CNOT}: \texttt{qbit}, \texttt{qbit}\\
&\textbf{equations}: q,q_1,q_2:\texttt{qbit}\\
&\qquad\quad \mathit{NOT}[q]\ket{0}_q=\ket{1}_q,\quad \mathit{NOT}[q]\ket{1}_q=\ket{0}_q\\ 
&\qquad\quad H[q]\ket{0}_q=\ket{+}_q,\quad H[q]\ket{1}_q=\ket{-}_q\\ 
&\qquad\quad \mathit{CNOT}[q_1,q_2]\ket{+}_{q_1}\ket{0}_{q_2}=\ket{\beta_{00}}_{q_1q_2},\quad \mathit{CNOT}[q_1,q_2]\ket{+}_{q_1}\ket{1}_{q_2}=\ket{\beta_{01}}_{q_1q_2}\\ 
&\qquad\quad \mathit{CNOT}[q_1,q_2]\ket{-}_{q_1}\ket{0}_{q_2}=\ket{\beta_{10}}_{q_1q_2},\quad \mathit{CNOT}[q_1,q_2]\ket{-}_{q_1}\ket{1}_{q_2}=\ket{\beta_{11}}_{q_1q_2}
\end{align*}

We can define a $\texttt{QBIT}$-algebra $\calA$ as follows:
\begin{itemize}
    \item Let $\calH_{\texttt{qbit}}^{\calA}=\calH_2$ be the $2$-dimensional Hilbert space.
    \item For quantum state symbols, define
    $\ket{b}^{\calA}=\ket{b}$ for $b\in \braces{0,1}$,
    $\ket{\pm}^{\calA}=\frac{1}{\sqrt{2}}(\ket{0}\pm\ket{1})$, and
    $\ket{\beta_{ab}}^{\calA}=\frac{1}{\sqrt{2}}
        \parens*{\ket{0}\ket{b}+(-1)^a\ket{1}\ket{1-b}}$
    for $a,b\in \braces{0,1}$.
    \item For quantum gate symbols, define
    \begin{equation*}
        \mathit{NOT}^{\calA}=X,\qquad
        H^{\calA}=H,\qquad
        \mathit{CNOT}^{\calA}=\mathrm{CNOT},
    \end{equation*}
    where $X$, $H$, and $\mathrm{CNOT}$ are the Pauli $X$, Hadamard, and controlled-NOT gates, respectively.
\end{itemize}
It is easy to check that $\calA$ satisfies all equations in $\texttt{QBIT}$, and hence $\calA$ is a $\texttt{QBIT}$-algebra.
\end{example}

\subsection{A Simple Application Example: Programming Quantum Query Algorithms}
\label{sub:oracle}

\subsubsection*{Quantum query algorithms} To further understand the role of ADTs in quantum computing, let us consider a large class of quantum algorithms, namely quantum query algorithms. Let $\overline{q_w}$ and $\overline{q_o}$ be two disjoint strings of qubits of lengths $l_w$ and $l_o$, called the work qubits and output qubits, respectively. Suppose we are given a quantum oracle $\calO$, which is a black-box unitary operator. Then a quantum query algorithm \cite{Amb} with queries to $\calO$ is usually described as a quantum circuit:
$$C=U_T\calO{}U_{T-1}\cdots U_2\calO{}U_1\calO{}U_0$$ where $U_t$ $(t=0,1,\ldots,T)$ are quantum gates independent of $\calO$. The circuit $C$ is initialised in state $\ket{0}_{\overline{q_w}}\ket{0}_{\overline{q_o}}$. Then the algorithm outputs $x$ (a bit string) with probability $$\mbox{Prob}(x)=\norm*{\bra{x}_{\overline{q_o}}C\ket{0}_{\overline{q_w}}\ket{0}_{\overline{q_o}}}^2.$$

\subsubsection*{Oracles as ADTs} Oracle synthesis functionality has already been integrated into quantum programming platforms like Qiskit \cite{qiskit} and Q\# \cite{sharp}. Ideally, a programmer does not need to know the implementation details of the oracle $\calO$ when programming the quantum query algorithm. Instead, it can be viewed through the lens of ADTs. Consider a quantum specification $(S_q,\calK,\calU,E)$, where $\texttt{qbit}\in S_q$, $\calO\in \calU$, and $\mathit{sort}(\calO)=\underbrace{\texttt{qbit},\ldots,\texttt{qbit}}_{l_w+l_o}$. The equations in $E$ may specify properties that the oracle $\calO$ should satisfy without fixing its implementation details. In this way, the quantum query algorithm can be programmed as:
\begin{align*}P\equiv\ &\overline{q_w}:= \ket{0}^{l_w}; \overline{q_o}:=\ket{0}^{l_o};\\ 
&\texttt{for}\ i\ \texttt{from}\ 0\ \texttt{to}\ T-1\ \texttt{do}\ U_i[\overline{q}];\\ &\qquad\qquad\qquad\qquad\qquad\ \calO[\overline{q}];\quad (\mbox{calling the oracle operation})\\
& U_T[\overline{q}];\\ & x:=M[\overline{q_o}]
\end{align*} where $\overline{q}=\overline{q_w}\cup\overline{q_o}$ and $M$ is the measurement in the computational basis of $\overline{q_o}$. 

\section{Quantisation of Data Types}\label{sec-quantisation}

The abstract notion of quantum ADTs was introduced in \Cref{sec-QDT}. In practical applications, a quantum ADT does not stand alone; rather, it often arises as an extension of a corresponding classical ADT. An application scenario is that, when developing a quantum program, the programmer realises that they need to apply a function, or more generally an ADT, that does not comply with certain constraints (e.g., non-unitarity) imposed by quantum mechanics. In such cases, the function or ADT must be quantised, and what the programmer actually invokes in the program is its quantisation. A typical example is a quantum oracle (see \Cref{sub:oracle}), which is often lifted from a black-box classical function.  
So, in this section, we consider how a classical data type defined in \Cref{sec-CDT} can be extended to a quantum data type introduced in \Cref{sec-QDT}. We call this extension process \textit{quantisation}.

In this section, we make the following conventions.
Given a classical signature, for each sort $s$, we assume a designated constant symbol $0_s:\rightarrow s$.
Given a quantum signature, for any sort $s$, we assume a designated quantum state symbol $\ket{0_s}:s$.
For a string $\overline{s}=s_1,\ldots,s_n$ of sorts, the quantum state $\ket{\overline{0}_{\overline{s}}}^{\calA}$ in a given algebra $\calA$ is then $\ket{0_{s_1}}^{\calA}\otimes \ldots\otimes \ket{0_{s_n}}^{\calA}$. When the context is clear, we omit the subscript $s$ and use $\ket{\overline{0}}$ to abbreviate $\ket{0} \ldots \ket{0}$. 
These designated symbols are implicitly assumed to be included in the signatures discussed below.
Given an algebra $A$ of a signature, for each sort $s$, we assume that $A_s$ is finite and there is a designated bijection
\begin{equation}\label{int-rep}
    k_s^A:A_s\rightarrow \Z_{\abs*{A_s}},
    \qquad k_s^A(0_s^A)=0
\end{equation} where for a positive integer $m$,  $\Z_m=\{0,1,...,m-1\}$. 
Intuitively, bijection $k_s(a)$ gives the default integer label for every element $a\in A_s$; that is, we use $k_s$ to enumerate all elements of $A_s$.

Let us first consider how a classical signature can be quantised to a quantum signature.

\begin{definition}[Quantisation of signature]
\label{def:quantisation-of-signature}
Let $\mathit{Sig}=(S,\mathit{OP})$ be a signature and $\mathit{QSig}=(S_q,\calK,\calU)$ a quantum signature.
A quantisation from $\mathit{Sig}$ to $\mathit{QSig}$ is a mapping $\eta$ such that \begin{enumerate}\item  
each sort $s\in S$ is mapped to $s_q=\eta(s)$; 
\item each constant symbol $c:\rightarrow s$ is mapped to a quantum state symbol $\ket{\psi_c}=\eta(c)$ of sort $s_q$. We additionally require $\eta(0_s)=\ket{0_{s_q}}$;
and \item each nonconstant operation symbol $O:s_1,\ldots,s_n\rightarrow s$ is mapped to a triple of quantum gate symbols $(U_O, \mathit{pre}_O, \mathit{post}_O)=\eta(O)$, where $\mathit{sort}(U_O)=\mathit{sort}(\mathit{pre}_O)=\mathit{sort}(\mathit{post}_O)\in S_q^+$ is a string of sorts in one of the following forms:
\begin{itemize}
    \item 
    (in-place) $s_{1q},\ldots,s_{nq}, \overline{s_a}$ for some $\overline{s_a}\in S_q^*$, with $s_q=s_{iq}$ for some $i$;
    \item 
    (out-of-place) $s_{1q},\ldots,s_{nq}, s_q,\overline{s_a}$ for some $\overline{s_a}\in S_q^*$.
\end{itemize}
\item 
    In the above, the mappings $s\mapsto s_q$, $c\mapsto \ket{\psi_c}$, and $O\mapsto U_O$ are injective.
\end{enumerate}
\end{definition}

We need to carefully explain the above condition~(3).
For a nonconstant operation $O$, it is natural to expect its quantisation lifts its behaviour on classical data to the quantum setting. However, there are multiple ways to represent classical data as quantum data,
e.g., encoding into computational-basis quantum states, encoding into relative phases, encoding into quantum amplitudes. Concrete examples can be found in \Cref{sub:bit-oracle,sub:phase-oracle,sec-app,sec-qec}.
Therefore, the quantum gate symbols $\mathit{pre}_O$ and $\mathit{post}_O$ in $\eta(O)$ serve for the representation of classical data in the quantum world, while $U_O$ serves for the lifted $O$ in the quantum world.
In particular, the quantum gate symbol $\mathit{pre}_O$ (resp.\ $\mathit{post}_O$) corresponds to how classical data is preprocessed into (resp.\ postprocessed from) quantum data. Then, $U_O$ corresponds to the quantisation of $O$ that has quantum input and output.
Their concrete interpretations in an algebra is given later.
The injectivity of $O\mapsto U_O$ in condition~(4) means $U_O$ is unique for each $O$, but $\mathit{pre}_O$ and $\mathit{post}_O$ can be the same for different $O$.

For convenience, if an operation symbol $O$ is quantised in the in-place form,
we assume the index $i$ for $s_q=s_{iq}$ is implicitly fixed and used consistently in the other quantisation notions.
The string $\overline{s_a}$ of sorts represents ancilla sorts needed for the quantised operation.

Next, we consider how an algebra can be quantised along the quantisation of the signature. %

\begin{definition}[Quantisation of algebra]
\label{def:quantisation-of-algebra}
Let $\eta$ be a quantisation from $\mathit{Sig}$ to $\mathit{QSig}$. Then a quantum algebra $\calA$ of signature $\mathit{QSig}$ is a quantisation of an algebra $A$ of signature $\mathit{Sig}$ along $\eta$ if: 
\begin{enumerate}\item for each sort $s\in S$, 
   $\braces*{\ket{a}\midv a\in A_s}$ forms an orthonormal basis of $\calH_{s_q}^{\calA}$; 
\item for each constant symbol $c:\rightarrow s$, $\ket{\psi_c}^\calA=\ket{c^A}$. In particular, $\ket{0_{s_q}}^{\calA}=\ket{0_s^A}$; and  
\item for each nonconstant operation symbol $O:s_1,\ldots,s_n\rightarrow s$, and for all $a_1\in A_{s_1},\ldots,a_n\in A_{s_n}$, we have:
\begin{equation}\label{eq:quantisation-condition} \norm*{\bra{O^A(a_1,\ldots,a_n)}\mathit{post}_O^{\calA}U_O^\calA \mathit{pre}_O^{\calA}\parens*{\ket{a_1}\cdots\ket{a_n}\ket{\overline{0}}}}^2=1.\end{equation}
Here, the subsystems involved in quantum states and gates are implicitly assumed to be ordered according to $\mathit{sort}(U_O)$. In particular,
\begin{itemize}
    \item   
    (in-place) $\ket{\overline{0}}$ denotes $\ket{\overline{0}_{\overline{s_a}}}^{\calA}$, and $\ket{O^A(a_1,\ldots,a_n)}$ is a state in $\calH_{s_{iq}}^{\calA}$, where $i$ is the fixed index in \Cref{def:quantisation-of-signature}; and
    \item
    (out-of-place) $\ket{\overline{0}}$ denotes $\ket{\overline{0}_{s_q,\overline{s_a}}}^{\calA}$, and $\ket{O^A(a_1,\ldots,a_n)}$ is a state in $\calH_{s_{q}}^{\calA}$.
\end{itemize}
\end{enumerate}
\end{definition}

Intuitively, in \Cref{eq:quantisation-condition}, $\mathit{pre}_O^{\calA}$ converts a computational-basis input (the default quantum representation of classical data) into a quantum input for $U_O^{\calA}$, and $\mathit{post}_O^{\calA}$ converts the quantum output of $U_O^{\calA}$ into a computational-basis output. 
The condition in \Cref{eq:quantisation-condition} says that, under the preprocessing $\mathit{pre}_O^{\calA}$ and postprocessing $\mathit{post}_O^{\calA}$, the behaviour of $U_O^{\calA}$ agrees with $O^A$ on all classical inputs (with probability $1$), so $U_O^{\calA}$ quantises $O^A$.
In \Cref{sub:bit-oracle,sub:phase-oracle,sec-app,sec-qec}, we will present illustrative examples,
where the utility of the preprocessing $\mathit{pre}_O$ and postprocessing $\mathit{post}_O$ becomes more concrete.

It remains to consider how a classical specification can be quantised. To this end, we need to introduce a quantisation of classical terms.

\begin{definition} [Quantisation of terms]
\label{def:quantisation-of-terms}
Let $\eta$ be a quantisation from $\mathit{Sig}$ to $\mathit{QSig}$. 
Suppose $\mathit{QSig}$ contains a quantum gate symbol $\mathit{SUM}_s:s_q,s_q$ for each sort $s$ in $\mathit{Sig}$, where $s_q=\eta(s)$.
Let $t$ be a term in $\mathit{Sig}$.
For each variable $x:s$ occurring in $t$, designate a distinct quantum variable $q_x:s_q$.

The quantisation $\eta(t)$ is defined recursively as follows:
\begin{enumerate}
    \item If $t=c:\rightarrow s$ is a constant symbol, choose a fresh variable $p:s_q$ as the principal variable and define $\eta(t)=\ket{\psi_c}_{p}$, where $\ket{\psi_c}=\eta(c)$;
    \item If $t=x$ is a variable of sort $s$, choose a fresh variable $q:s_q$ as the principal variable and define $\eta(t)=\mathit{SUM}_{s}[q_x,q]\ket{0}_{q}$; and
    \item Suppose $t=O(t_1,\ldots,t_n)$ is a term of sort $s$ and $\eta(t_i)=C_i\ket{\psi_i}$ for $i=1,...,n$, where each $\eta(t_j)$ has principal variable $p_j$ of sort $s_{jq}$.
    Let $\overline{p}=p_1,\ldots,p_n$. By renaming variables if necessary, assume that, apart from the designated $q_x$, all other quantum variables in $\eta(t_1),\ldots,\eta(t_n)$ are distinct. 

    We define
    $$\eta(t)=C_1;\ldots ; C_n;\mathit{pre}_O[\overline{p},\overline{r}]; U_O[\overline{p},\overline{r}];\mathit{post}_O[\overline{p},\overline{r}]\ket{\psi_1} \ldots \ket{\psi_n}\ket{\overline{0}}_{\overline{r}}.$$
    Here $\overline{r}$ is a sequence of fresh variables such that
    \begin{itemize}
        \item
            (in-place) 
            $\mathit{sort}(\overline{r})=\overline{s_a}$ and $p_i$ is chosen to be the principal variable, where $i$ is the fixed index in \Cref{def:quantisation-of-signature}; 
        \item
            (out-of-place) 
            $\mathit{sort}(\overline{r})=s_q,\overline{s_a}$ and $\overline{r}=r_0,\overline{r_a}$, where $r_0$ is chosen to be the principal variable.
    \end{itemize}
\end{enumerate}
\end{definition}

In \Cref{def:quantisation-of-terms}, note that the same $q_x$ is used for every occurrence of variable $x$ in $t$.
In clause (2) of this definition, the design $\eta(t)=\mathit{SUM}_{s}[q_x,q]\ket{0}_{q}$ is subtle and needs to be carefully explained.
Intuitively, if the gate symbol $\mathit{SUM}$ is interpreted properly in an algebra,
$\mathit{SUM}_{s}[q_x,q]\ket{0}_{q}$ copies the computational-basis state of quantum variable $q_x$ to $q$, which is initialised in state $\ket{0}_q$.
It has two purposes. First, copying in the computational basis reversibly handles the multiple occurrences of variable $x$ in term $t$, where $q_x$ is used to store the input of $x$ and each fresh $q$ corresponds to a single occurrence of $x$.
Second, since $q_x$ are not chosen to be principal variables, tracing out $q_x$ at the end ``dephases'' each corresponding $q$ in the computational basis. This is important for the correctness of \Cref{thm:quantise}, when we transition from a classical equation with a universal quantifier $(\forall \overline{x})$  over classical inputs to a quantum equation with quantifier $(\forall \overline{q})$ over quantum inputs.

With quantisations of classical terms, we can further define quantisations of equations in a classical specification.

\begin{definition} [Quantisation of equations]
\label{def:quantisation-of-equations}
Let $\eta$ be a quantisation from $\mathit{Sig}$ to $\mathit{QSig}$, and let $e=(\forall\overline{x}:\overline{s})(t_1=t_2)$ be a classical equation, where $\overline{x}=x_1,\ldots,x_n$ is a string of all variables occurring in $t_1$ or $t_2$. We assume that the designated variables $\overline{q}=q_{x_1},\ldots,q_{x_n}$ (see \Cref{def:quantisation-of-terms}) in $t_1$ and $t_2$ are aligned, and $\eta(t_1)$ and $\eta(t_2)$ have the same principal variable, say $p$. %
Then, the quantisation $\eta(e)$ of $e$ along $\eta$ is defined as follows: 
\begin{equation*}
    \eta(e)=(\forall\overline{q}:\overline{s_q})(\eta(t_1)=\eta(t_2))\Downarrow p,
\end{equation*}
where $\overline{s_q}$ is the quantisation of $\overline{s}$.
\end{definition}

With the above preparation, we are able to present the main result of this section, which describes a way of deriving a specification of a quantised ADT by lifting from the classical specification. To describe this lifting, for $a,b\in A_s$, we define
\begin{equation}
    \label{eq:boxplus_def}
    a\boxplus_s^A b=(k_s^A)^{-1}\parens*{k_s^A(a)+k_s^A(b)\bmod \abs*{A_s}}
\end{equation} where $k_s^A$ is given in equation (\ref{int-rep}). 
Intuitively, the operator $\boxplus_s^A$ naturally extends the binary XOR $\oplus$ operator up to the bijection $k_s^A$. In particular, in \Cref{eq:boxplus_def}, the inputs $a,b$ are mapped to the integer domain $\Z_{\abs{A_s}}$, added in $\Z_{\abs{A_s}}$, and then mapped back to the domain $A_s$. It is easy to verify that $a\boxplus_s^A 0_s^A= 0_s^A \boxplus_s^A a=a$ for any $a$.
When the algebra $A$ is clear, we simply write $k_s$ and $\boxplus_s$ for $k_s^A$ and $\boxplus_s^A$, respectively. 

\begin{theorem}[Lifting of specification]
    \label{thm:quantise}
Let $\mathit{Spec}=(\mathit{Sig},E)$ be a classical specification, and let $\eta$ be a quantisation from $\mathit{Sig}$ to $\mathit{QSig}$. 
Suppose $\mathit{QSig}$ contains a quantum gate symbol $\mathit{SUM}_s:s_q,s_q$ for each sort $s$ in $\mathit{Sig}$, where $s_q=\eta(s)$.
We define $\mathit{QSpec}=(\mathit{QSig},\eta(E))$, where $\eta(E)=\braces*{\eta(e)\midv\mbox{equation}\ e\in E}$. If
\begin{enumerate}
    \item $A$ is an algebra of $\mathit{Spec}$; and 
    \item $\calA$ is a quantisation of $A$ along $\eta$ and $\mathit{SUM}^{\calA}_s:\ket{a}\ket{b}\mapsto \ket{a} \ket{a\boxplus_s b}$ for any sort $s$ and any $a,b\in A_s$,
\end{enumerate} 
then $\calA$ is a quantum algebra of $\mathit{QSpec}$.
\end{theorem}

Here, corresponding to the discussion below \Cref{def:quantisation-of-terms}, the unitary operator $\mathit{SUM}^{\calA}_s:\ket{a}\ket{b}\mapsto \ket{a} \ket{a\boxplus_s b}$ is a canonical way to realise the copying in the computational basis when $b=0$.
The proof of \Cref{thm:quantise} is deferred to \Cref{appendix}.

The above definitions of quantisation are general.
In the following subsections, we examine two specific quantisation methods that appear commonly in quantum algorithm design.  

\subsection{Bit Oracle Quantisation}
\label{sub:bit-oracle}

The first quantisation method is bit oracle quantisation (BO-quantisation for short).
The BO-quantisation of a classical algebra is defined as follows.

\begin{definition}[BO-quantisation]\label{def:bo-quantisation} Let $\eta$ be a quantisation from $\mathit{Sig}$ to $\mathit{QSig}$ and $A$ be a $\mathit{Sig}$-algebra. Then, a quantisation $\calA$ of $A$ along $\eta$ is a BO-quantisation if
for each nonconstant operation symbol $O:s_1,\ldots,s_n\rightarrow s$:
\begin{enumerate}
    \item 
    $U_O$ is out-of-place with $\overline{s_a}=\emptyset$;
    \item 
    $\mathit{pre}^{\calA}_O=\mathit{post}^{\calA}_O=\Id_{\mathit{sort}(U_O)}$; and 
    \item
    for any $a_1\in A_{s_1},\ldots,a_n\in A_{s_n},a\in A_s$,
    \begin{equation}\label{eq:bo-oracle} 
        U_O^\calA(\ket{a_1}\cdots\ket{a_n}\ket{a})=\ket{a_1}\cdots\ket{a_n}\ket{a\boxplus_s O^A(a_1,\ldots,a_n)}.
    \end{equation}
\end{enumerate}
\end{definition}

It is easy to check that $U_O^\mathcal{A}$ defined in equation (\ref{eq:bo-oracle}) indeed satisfies the condition (\ref{eq:quantisation-condition}). 
When $\mathit{QSig}$ includes the gate symbol $\mathit{SUM}_s$ for each sort $s$, the BO-quantisation of terms and equations can be naturally induced by \Cref{def:bo-quantisation}, according to \Cref{def:quantisation-of-signature,def:quantisation-of-equations}.

\Cref{def:bo-quantisation} naturally yields the common bit oracle in quantum algorithms as follows,
where the designated $k_{\texttt{bit}}$ is the identity and $\boxplus_{\texttt{bit}}$ reduces to XOR.

\begin{example}
    \label{exp:bit-oracle}
Consider the following
\begin{itemize}
    \item
        Signature $\mathit{Sig}=(S,\mathit{OP})$ with $S=\braces*{\texttt{bit}}$ and $\mathit{OP}=\braces*{0_{\texttt{bit}},f}$,
        where $0_{\texttt{bit}}:\rightarrow \texttt{bit}$ and $f:\underbrace{\texttt{bit},\ldots,\texttt{bit}}_n\rightarrow \texttt{bit}$.
    \item
        Algebra $A$ of signature $\mathit{Sig}$ with $A_{\texttt{bit}}=\braces*{0,1}$ and $0_{\texttt{bit}}^A=0$, where $f^A:\braces{0,1}^n\rightarrow\braces{0,1}$ is a $n$-ary Boolean function.
\end{itemize}
Then the following gives a BO-quantisation of $A$:
\begin{itemize}
    \item
        Signature $\mathit{QSig}=(S_q,\calK,\calU)$ with $S_q=\braces*{\texttt{qbit}}$, $\calK=\braces*{\ket{0_{\texttt{qbit}}}}$, and $\calU=\braces*{U_f,\mathit{pre}_f,\mathit{post}_f}$,
        where $\ket{0_{\texttt{qbit}}}: \texttt{qbit}$ and $U_f,\mathit{pre}_f,\mathit{post}_f:\underbrace{\texttt{qbit},\ldots,\texttt{qbit}}_{n+1}$.
    \item
        Algebra $\calA$ of signature $\mathit{QSig}$: $\calH_{\texttt{qbit}}^{\calA}=\Co^2$, $\ket{0_{\texttt{qbit}}}^{\calA}=\ket{0}$, 
        $\mathit{pre}_f^{\calA}=\mathit{post}_f^{\calA}=\Id$,
        and
        \begin{equation}\label{eqn-BO}
            U_f^{\mathcal{A}}\ket{\overline{x},b}=\ket{\overline{x},b\oplus f^A(\overline{x})}
        \end{equation}
        for any $\overline{x}\in\braces{0,1}^n$ and $b\in\braces{0,1}$. 
\end{itemize}
\end{example}

In this example, the unitary operator $U_f^{\calA}$ is often known as the bit oracle or standard oracle used by a quantum algorithm to query a black-box function $f$.
The bit oracle has been employed in the design of many quantum algorithms, e.g., the Deutsch-Jozsa algorithm~\cite{DeutschJozsa} and the modular exponentiation subroutine of Shor's algorithm~\cite{Shor}.
It should be noted that on the RHS of \Cref{eqn-BO}, the first $n$ qubits are used to record the input $\overline{x}$, and the last qubit is used to store the computational outcome $f(\overline{x})$. In this way, $U_f$ is made invertible and thus a unitary operator.

\subsection{Phase Oracle Quantisation}
\label{sub:phase-oracle}

The second quantisation method is phase oracle quantisation (PO-quantisation for short). 
The PO-quantisation of a classical algebra is defined as follows.

\begin{definition}[PO-quantisation] 
    \label{def:po-quantisation}
    Let $\eta$ be a quantisation from $\mathit{Sig}$ to $\mathit{QSig}$ and $A$ be a $\mathit{Sig}$-algebra. Then, a quantisation $\calA$ of $A$ along $\eta$ is a PO-quantisation if for each nonconstant operation symbol $O:s_1,\ldots,s_n\rightarrow s$:
\begin{enumerate}
    \item
    $U_O$ is out-of-place with $\overline{s_a}=\emptyset$.
    \item
    $\mathit{pre}^{\calA}_O=\Id_{s_{1q},\ldots,s_{nq}}\otimes \mathit{QFT}_{\abs{A_s}}$
    and $\mathit{post}^{\calA}_O=\Id_{s_{1q},\ldots,s_{nq}}\otimes \mathit{QFT}_{\abs{A_s}}^\dagger$,
    where $\mathit{QFT}_{\abs{A_s}}$ is the quantum Fourier transform on $\Z_{\abs{A_s}}$ defined by
    \begin{equation*}
       \mathit{QFT}_{\abs{A_s}}\ket{a}=\frac{1}{\sqrt{\abs{A_s}}}\sum_{b\in A_s} \omega^{k_s(a)\cdot k_s(b)}\ket{b}
    \end{equation*}
    for any $a\in A_s$,
    where $\omega=e^{2\pi i/\abs{A_s}}$; and
    \item 
     for any $a_1\in A_{s_1},\ldots,a_n\in A_{s_n},a\in A_s$,
    \begin{equation}\label{eq:po-oracle} 
       U_O^\calA\ket{a_1,\ldots, a_n} \ket{a}=\omega^{k_s(a)\cdot k_s(O^A(a_1,\ldots,a_n))}\ket{a_1,\ldots, a_n} \ket{a}.
    \end{equation}
\end{enumerate}
\end{definition}

It is also easy to check that $U_O^\mathcal{A}$ defined in equation (\ref{eq:po-oracle})  satisfies the condition (\ref{eq:quantisation-condition}). When $\mathit{QSig}$ includes the gate symbol $\mathit{SUM}_s$ for each sort $s$, the PO-quantisation of terms and equations can be naturally induced by \Cref{def:po-quantisation}, according to \Cref{def:quantisation-of-signature,def:quantisation-of-equations}.

\Cref{def:po-quantisation} naturally yields the common phase oracle (or more accurately, the controlled version of the common phase oracle) in quantum algorithms,
where the designated $k_{\texttt{bit}}$ is the identity, $\omega=-1$ and $\mathit{QFT}_2=H$ is the Hadamard operator.

\begin{example}
    Consider the same signature $\mathit{Sig}$ and algebra $A$ as in \Cref{exp:bit-oracle}.
    Then, the following gives a PO-quantisation of $A$:
    \begin{itemize}
        \item
            Signature $\mathit{QSig}=(S_q,\calK,\calU)$: $S_q=\braces*{\texttt{qbit}}$, $\calK=\braces*{\ket{0_{\texttt{qbit}}}}$, and $\calU=\braces*{U_f,\mathit{pre}_f,\mathit{post}_f}$,
            where $\ket{0_{\texttt{qbit}}}:\texttt{qbit}$ and $U_f,\mathit{pre}_f,\mathit{post}_f:\underbrace{\texttt{qbit},\ldots,\texttt{qbit}}_{n+1}$.
        \item
            Algebra $\calA$ of signature $\mathit{QSig}$: $\calH_{\texttt{qbit}}^{\calA}=\Co^2$, $\ket{0_{\texttt{qbit}}}^{\calA}=\ket{0}$, 
            $\mathit{pre}_f^{\calA}=\mathit{post}_f^{\calA}=\Id_n\otimes H$, and
            \begin{equation*}
                U_f^{\mathcal{A}}\ket{\overline{x},b}=(-1)^{b\cdot f^A(\overline{x})}\ket{\overline{x},b}
            \end{equation*}
            for any $\overline{x}\in\braces{0,1}^n$ and $b\in\braces{0,1}$.
    \end{itemize}
\end{example}

In this example, $U_f$ is often known as the (controlled) phase oracle for a quantum algorithm to query a black-box function $f$.
The phase oracle has been applied in Grover's search~\cite{Grover}, amplitude amplification~\cite{Amplitude}, and several other quantum algorithms. 

\subsection{Phase Kickback Trick}

In quantum algorithms, the phase kickback trick provides a useful connection between a bit oracle and a phase oracle,
and now we can view it through the lens of ADTs.
For example, in the Deutsch-Jozsa algorithm that decides whether a given Boolean function $f:\braces{0,1}^n\rightarrow\braces{0,1}$ is constant or balanced, the bit oracle $U_f$ is defined by $$U_f\ket{\overline{x},y}=\ket{\overline{x},y\oplus f(\overline{x})}$$ for any $\overline{x}\in\braces{0,1}^n$ and $y\in\braces{0,1}$. A key idea in this algorithm is that the target register is initialised in the superposition $\ket{-}=\frac{1}{\sqrt{2}}(\ket{0}-\ket{1})$. We observe:
\begin{equation}\label{kick}U_f\ket{\overline{x},-}=\ket{\overline{x}}\otimes (-1)^{f(\overline{x})}\ket{-}=(-1)^{f(\overline{x})}\ket{\overline{x}}\otimes \ket{-}.\end{equation}
In the final expression above, the phase $(-1)^{f(\overline{x})}$ appears in front of the quantum state. 
On the other hand, if the target register is initialised in $\ket{+}=\frac{1}{\sqrt{2}}(\ket{0}+\ket{1})$, then
\begin{equation}
    \label{kick-2}
    U_f\ket{\overline{x},+}=\ket{\overline{x}}\otimes \ket{+}.
\end{equation}
Combining \Cref{kick,kick-2} gives
\begin{equation}
    \label{eq:kick-1-bit}
   (\Id_n\otimes H) U_f(\Id_n\otimes H)\ket{\overline{x},b}=(-1)^{b\cdot f(\overline{x})}\ket{\overline{x},b}.
\end{equation}

More generally, the phase kickback implies a correspondence between BO-quantisation and PO-quantisation.
Suppose $\eta$ is a quantisation from $\mathit{Sig}$ to $\mathit{QSig}$.
Let $\calA$ and $\calB$ be BO- and PO-quantisations of the same algebra $A$ along $\eta$, respectively.
For each sort $s$ in $\mathit{Sig}$, suppose that $\calA$ and $\calB$ have the same quantum domain 
\begin{equation*}
    \calH_{s_q}^{\calA}=\calH_{s_q}^{\calB}=\Span\braces*{\ket{a}:a\in A_s}.
\end{equation*}
Then, analogous to \Cref{eq:kick-1-bit}, for any operation $O:s_1,\ldots,s_n\rightarrow s$, we have:
\begin{equation}
    \label{eq:phase-kickback}
    U_O^{\calB}=\parens*{\Id_{s_{1q},\ldots,s_{nq}}\otimes \mathit{QFT}_{\abs*{A_s}}}U_O^{\calA}
    \parens*{\Id_{s_{1q},\ldots,s_{nq}}\otimes \mathit{QFT}_{\abs*{A_s}}^\dagger},
\end{equation}
which is essentially equivalent to 
\begin{equation*}
    \mathit{post}_O^{\calA}U_O^{\calA}\mathit{pre}_O^{\calA}=\mathit{post}_O^{\calB}U_O^{\calB}\mathit{pre}_O^{\calB}.
\end{equation*} 
This means that BO-quantisation and PO-quantisation can be related by changing how classical data are encoded into quantum data through $\mathit{pre}_O$ and $\mathit{post}_O$.

\section{Illustrative Application I: Quantum Arrays}\label{sec-app}

Arrays are a fundamental data type widely used in classical programming. To illustrate the framework developed in \Cref{sec-QDT,sec-quantisation}, we show how quantum arrays can be defined via the quantisation of classical arrays. In particular, this approach offers a new perspective for understanding Quantum Random Access Memories (QRAMs)~\cite{QRAM,QRAMArch}.

\subsection{Classical Arrays}\label{sec-6150148}

For convenience of the reader, let us recall the basic theory of arrays. First, we describe the signature $\mathit{Sig}=(S,\mathit{OP})$ of the ADT of arrays. Let $\mathcal{I}$ be a set of index sorts and $\calV$ a set of value sorts, so that $S=\mathcal{I}\cup\calV$. The value sorts and their dimensions are generated as follows:
\begin{enumerate}
\item $\texttt{val}\in\calV$ is a primitive value sort, and $\dim(\texttt{val})=0$;
\item if $I\in\mathcal{I}$ and $\tau\in\calV$, then $\texttt{array}(I,\tau)\in\calV$, and $\dim(\texttt{array}(I,\tau))=\dim(\tau)+1$.
\end{enumerate}
A canonical algebra $A$ of $\mathit{Sig}$ gives the following interpretation. For every $I\in\mathcal{I}$ and $\tau\in\calV$, define
\begin{equation*}
    A_{\texttt{array}(I,\tau)}=A_I\rightarrow A_\tau.
\end{equation*}
Consequently, if
\begin{equation*}
    \tau=\texttt{array}(I_1,\texttt{array}(I_2,\ldots,\texttt{array}(I_n,\texttt{val})\ldots)),
\end{equation*}
then
\begin{equation*}
    A_\tau=A_{I_1}\rightarrow\parens*{A_{I_2}\rightarrow\cdots\rightarrow\parens*{A_{I_n}\rightarrow A_{\texttt{val}}}\cdots}.
\end{equation*}
Equivalently, $A_\tau=(A_{I_1}\times\cdots\times A_{I_n})\rightarrow A_{\texttt{val}}$.

For the operation symbols, let $\mathit{OP}$ contain two classes of symbols:
\begin{equation}
    \label{sig-array}
    \begin{split}
        &\mathit{read}_{I,\tau}:\texttt{array}(I,\tau),I\rightarrow\tau\\
        &\mathit{write}_{I,\tau}:\texttt{array}(I,\tau),I,\tau\rightarrow\texttt{array}(I,\tau)
    \end{split}
\end{equation} 
where $I\in\mathcal{I}$ and $\tau\in\calV$.
They are also called \textit{selector} and \textit{store} operations, respectively. For simplicity, the subscripts $I,\tau$ of $\mathit{read}_{I,\tau}$ and $\mathit{write}_{I,\tau}$ will often be omitted.
In the canonical algebra $A$, they are interpreted as follows:
\begin{equation}
    \label{op-array}
    \begin{split}
        \mathit{read}_{I,\tau}^A &: (a,j)\mapsto a[j]\\
        \mathit{write}_{I,\tau}^A &: (a,j,v)\mapsto a'
    \end{split}
\end{equation}
for $a\in A_{\texttt{array}(I,\tau)}$, $j\in A_I$, and $v\in A_\tau$,
where $a'[i]=a[i]$ if $i\neq j$ and $a'[j]=v$.
Intuitively, $\mathit{read}(a,j)$ denotes the value stored in array $a$ at index $j$, and $\mathit{write}(a,j,v)$ denotes the array obtained by updating $a$ with value $v$ stored at index $j$. 

We assume countably many variables and constants of each sort. The terms are constructed from variables and constants using operation symbols, and their sorts are defined in the usual way. %
A basic theory of arrays was first introduced by McCarthy \cite{Mc}, where the functions $\mathit{read}$ and $\mathit{write}$ are characterised by the following axioms.
For every $I\in\mathcal{I}$ and $\tau\in\calV$:
\begin{itemize}\item \textbf{Axiom 1 (read-over-write)}:
\begin{align*}\forall a:\texttt{array}(I,\tau), i,j:I,v:\tau.&(i=j\rightarrow\mathit{read}(\mathit{write}(a,i,v),j)=v)\\ &(i\neq j\rightarrow \mathit{read}(\mathit{write}(a,i,v),j)=\mathit{read}(a,j))\end{align*}
\item \textbf{Axiom 2 (extensionality)}:
\begin{align*}\forall a,b:\texttt{array}(I,\tau).(\forall i:I.\mathit{read}(a,i)=\mathit{read}(b,i))\rightarrow a=b
\end{align*}
\end{itemize} It should be pointed out that, essentially, the above two axioms do not constitute an equational specification as in \Cref{def-spec}. Instead, they are formulated in a first-order logical language. However, they are easy to understand for our purposes.

\subsection{Quantum Arrays}

Now let us examine quantum arrays through the lens of the quantisation of the classical ADT of arrays (see \Cref{sec-quantisation}).
Suppose $\mathit{QSig}=(S_q,\calK,\calU)$ is the quantum signature obtained from a quantisation $\eta$ of the classical signature for arrays. For $I\in\mathcal{I}$ and $\tau\in\calV$, write
\begin{equation*}
    I_q=\eta(I),\qquad \tau_q=\eta(\tau),\qquad
    \texttt{qarray}(I_q,\tau_q)=\eta(\texttt{array}(I,\tau)).
\end{equation*}
In particular, write $\texttt{qval}=\eta(\texttt{val})$. Then
\begin{equation*}
    S_q=\mathcal{I}_q\cup\calV_q,\qquad
    \mathcal{I}_q=\braces*{\eta(I):I\in\mathcal{I}},\qquad
    \calV_q=\braces*{\eta(\tau):\tau\in\calV}.
\end{equation*}
Let $A$ be the canonical algebra for the signature of classical arrays considered previously. We consider the following quantised algebra $\calA$ of $A$ along $\eta$. We define:
\begin{enumerate}
    \item
        $\calH_{\texttt{qval}}^{\calA}=\Span\braces*{\ket{v}:v\in A_{\texttt{val}}}$;
    \item
        for every $I\in\mathcal{I}$, $\calH_{I_q}^{\calA}=\Span\braces*{\ket{j}:j\in A_I}$; and
    \item
        for every $\tau\in\calV$, $\calH_{\tau_q}^{\calA}=\Span\braces*{\ket{t}:t\in A_\tau}$. Furthermore, for every $I\in\mathcal{I}$ and $\tau\in\calV$,
        \begin{equation}
            \label{eq:qarray-rec}
            \calH_{\texttt{qarray}(I_q,\tau_q)}^{\calA}=
            \Span\braces*{\ket{a}:a\in A_{\texttt{array}(I,\tau)}}=
            \bigotimes_{j\in A_I} \calH_{\tau_q}^{\calA},
        \end{equation}
        where we denote $\ket{a}=\bigotimes_{j\in A_I}\ket{a[j]}$, and the indices $j$ are ordered by keeping the designated $k_I(j)$ increasing (see \Cref{sec-quantisation}).
        By expanding the RHS of \Cref{eq:qarray-rec} recursively, if
        \begin{equation*}
            \tau=\texttt{array}(I_1,\texttt{array}(I_2,\ldots,\texttt{array}(I_n,\texttt{val})\ldots)),
        \end{equation*}
        then
        \begin{equation*}
            \calH_{\tau_q}^{\calA}=
            \bigotimes_{(j_1,\ldots,j_n)\in A_{I_1}\times\cdots\times A_{I_n}}
            \calH_{\texttt{qval}}^{\calA}.
        \end{equation*}
\end{enumerate}
Next, we explain how the $\mathit{read}$ and $\mathit{write}$ symbols in $\mathit{OP}$ are quantised to gate symbols in $\calU$. As in \Cref{sec-quantisation}, for each $I\in\mathcal{I}$ and $\tau\in\calV$,
we assume that the designated $0$ satisfies $0^{A}_{\texttt{array}(I,\tau)}[j]=0_{\tau}^A$ for any $j\in A_I$.
Let $k_I$ and $k_\tau$ be the designated bijections for $A_I$ and $A_\tau$, respectively. We recursively choose the designated bijection for the corresponding array sort to be
\begin{equation}
    \label{eq:k-array}    
    k_{\texttt{array}(I,\tau)}(a)
    =\sum_{j\in A_I}\abs*{A_\tau}^{k_I(j)}\cdot k_\tau(a[j]).
\end{equation}
For the concrete one-dimensional QRAM examples below, fix an index sort $\texttt{idx}\in\mathcal{I}$ and write
\begin{equation*}
    \texttt{qidx}=\eta(\texttt{idx}),\qquad
    \texttt{array}=\texttt{array}(\texttt{idx},\texttt{val}),\qquad
    \texttt{qarray}=\texttt{qarray}(\texttt{qidx},\texttt{qval}).
\end{equation*}

\subsubsection*{Quantum read}  
The operation $\mathit{read}_{I,\tau}^A$ can be quantised in the same way as the BO-quantisation in \Cref{def:bo-quantisation}:
\begin{enumerate}
    \item
        $U_{\mathit{read}}$, $\mathit{pre}_{\mathit{read}}$, and $\mathit{post}_{\mathit{read}}$ have the sort $\texttt{qarray}(I_q,\tau_q),I_q,\tau_q$.
    \item
        $\mathit{pre}_{\mathit{read}}^{\calA}=\mathit{post}_{\mathit{read}}^{\calA}=\Id$.
    \item
        $U_{\mathit{read}}^{\calA}\ket{a}\ket{j}\ket{v}=\ket{a}\ket{j}\ket{v\boxplus_\tau a[j]}$ for any $a\in A_{\texttt{array}(I,\tau)}$, $j\in A_I$, and $v\in A_\tau$.
\end{enumerate}
We simply use $\mathit{qread}$ to denote the gate $U_{\mathit{read}}$ for convenience.

What $\mathit{qread}^{\calA}$ effectively implements is a QRAM read operation.
Consider the case when $A_{\texttt{idx}}=\Z_N$, $A_{\texttt{val}}=\Z_M$, and their designated bijections are the identity maps.
Then, an array $a$ of sort $\texttt{array}$ stores $N$ classical values $a[0],a[1],\ldots,a[N-1]$ from $A_{\texttt{val}}$ at addresses $0,1,\ldots,N-1$.
In classical computing, RAM (Random Access Memory) allows one to access a memory cell by providing an address. In QRAM (Quantum Random Access Memory), the key idea is that the address itself can be in a quantum superposition. This means one can query many memory locations simultaneously in superposition. A QRAM read implements the following transformation:
\begin{equation}\label{QRM1}
    \sum_j\beta_j\ket{a[0],\ldots,a[j],\ldots,a[N-1]}\ket{j}\ket{0}\mapsto\sum_j\beta_j\ket{a[0],\ldots,a[j],\ldots,a[N-1]}\ket{j}\ket{a[j]},
\end{equation}
which is an instantiation of $\mathit{qread}^{\calA}$.
So, in one operation, we can ``load'' data from multiple addresses coherently.

\subsubsection*{Quantum write} 
The operation $\mathit{write}_{I,\tau}^A$ can be quantised in-place as follows:
\begin{itemize}
    \item
        $U_{\mathit{write}}$, $\mathit{pre}_{\mathit{write}}$, and $\mathit{post}_{\mathit{write}}$ have the sort $\texttt{qarray}(I_q,\tau_q),I_q,\tau_q$.
    \item
        $\mathit{pre}_{\mathit{write}}^{\calA}=\mathit{post}_{\mathit{write}}^{\calA}=\Id$.
    \item
        $U_{\mathit{write}}^{\calA}\ket{a}\ket{j}\ket{v}=\ket{a'}\ket{j}\ket{a[j]}$ for any $a\in A_{\texttt{array}(I,\tau)}$, $j\in A_I$, and $v\in A_\tau$, where $a'[i]=a[i]$ if $i\neq j$ and $a'[j]=v$.
\end{itemize}
We simply use $\mathit{qwrite}$ to denote the gate $U_{\mathit{write}}$ for convenience.

What $\mathit{qwrite}^{\calA}$ effectively implements is a QRAM swap operation.
Again, consider the case when $A_{\texttt{idx}}=\Z_N$, $A_{\texttt{val}}=\Z_M$, and their designated bijections are the identity maps.
Then, an array $a$ of sort $\texttt{array}$ stores $N$ classical values $a[0],a[1],\ldots,a[N-1]$ from $A_{\texttt{val}}$ at addresses $0,1,\ldots,N-1$.
A QRAM swap implements the following transformation:
\begin{equation}
    \sum_j\beta_j\ket{a[0],\ldots,a[j],\ldots,a[N-1]}\ket{j}\ket{v}\mapsto\sum_j\beta_j\ket{a[0],\ldots,v,\ldots,a[N-1]}\ket{j}\ket{a[j]}.
\end{equation} 
Therefore, in one operation, $\mathit{qwrite}$ can ``write'' data to multiple addresses coherently.

According to \Cref{thm:quantise}, by properly reexpressing \textbf{Axiom 1} and \textbf{Axiom 2} using equations we can derive specifications for the ADT of quantum arrays.
Beyond that, quantum arrays actually satisfy extra properties.
The following proposition shows that quantum array operations \textit{qread} and \textit{qwrite} satisfy an equation similar to the first part of the \textbf{Axiom 1} for classical array operations \textit{read} and \textit{write}. This property is specified as a quantum equation (see \Cref{def-qeq}). 
Suppose that for each sort $s$ in $\mathit{Sig}$,
$\mathit{SUM}_s:s_q,s_q$ is a gate symbol in $\mathit{QSig}$ and $\mathit{SUM}_s^{\calA}\ket{a}\ket{b}=\ket{a}\ket{a\boxplus_s b}$ for any $a,b\in A_s$.

\begin{proposition}\label{prop-array} The algebra of quantum arrays with gates $\mathit{qread}$ and $\mathit{qwrite}$ satisfies the following specification:
\begin{equation}
    \label{eq:qarray-spec}
    \begin{split}
&\forall q_a:\textnormal{\texttt{qarray}},
q_i:\textnormal{\texttt{qidx}},
q_v,q_u:\textnormal{\texttt{qval}}\\
        &\qquad\quad \parens*{\mathit{qwrite}[q_a,q_i,q_u];\mathit{qread}[q_a,q_i,q_v]}=(\mathit{SUM}[q_u,q_v];\mathit{qwrite}[q_a,q_i,q_u])
    \end{split}
\end{equation} 
\end{proposition}

However, a natural quantum analogue of \textbf{Axiom 2} is not true.
Let us consider the simplest case when $A_{\texttt{idx}}=A_{\texttt{val}}=\Z_2$.
Even if we restrict attention to the principal variable, the equation $$\parens*{\mathit{qread}[q_a,q_i,q_v]\ket{\Phi}=\mathit{qread}[q_a,q_i,q_v]\ket{\Psi}}\Downarrow q_v$$ does not imply $\ket{\Phi}=\ket{\Psi}$, because the state of a composite quantum system is not uniquely determined by the states of its subsystems, which is indeed one of the fundamental differences between classical and quantum systems.

\section{Illustrative Application II: Quantum Error Correction}
\label{sec-qec}

Quantum error correction is fundamental to fault-tolerant quantum computing~\cite{KnillLaflamme,gottesman2024surviving}. To further illustrate the framework developed in \Cref{sec-QDT,sec-quantisation}, we show how a quantum error-correcting code can be described as a quantisation of a classical error-correcting code. In particular, we first consider the quantisation of a classical $3$-bit repetition code and then show how its $9$-bit extension gives rise to the Shor code.

\subsection{Classical Repetition Code}
\label{sub:classical-repetition-code}

We first describe a classical repetition code as an ADT.
Consider the signature $\mathit{Sig}=(S,\mathit{OP})$, where $S=\braces*{\texttt{lbit},\texttt{pbit},\texttt{idx},\texttt{array}}$,  $\texttt{lbit}$ and $\texttt{pbit}$ denote sorts of logical bits and physical bits, respectively.
Following \Cref{sec-6150148}, $\texttt{idx}$ is an index sort, and $\texttt{array}=\texttt{array}(\texttt{idx},\texttt{pbit})$ is the sort for an array of physical bits. For simplicity, consider the canonical algebra $A$ of $\mathit{Sig}$ for the $3$-bit repetition code.
It gives:
\begin{equation*}
    A_{\texttt{lbit}}=A_{\texttt{pbit}}=\braces*{0,1},
    \qquad A_{\texttt{idx}}=\braces*{0,1,2},
    \qquad A_{\texttt{array}}=A_{\texttt{idx}}\rightarrow A_{\texttt{pbit}}.
\end{equation*}
Thus, $A_{\texttt{array}}$ consists of arrays of $3$ physical bits indexed by $\braces{0,1,2}$. For the operation symbols, let $\mathit{OP}$ contain the following symbols:
\begin{align*}
    \mathit{enc}&:\texttt{lbit}\rightarrow\texttt{array},\\
    \mathit{flip}&:\texttt{array},\texttt{idx}\rightarrow\texttt{array},\\
    \mathit{corr}&:\texttt{array}\rightarrow\texttt{array}.
\end{align*}
For $a=a[0],a[1],a[2]\in A_{\texttt{array}}$, define
\begin{equation*}
    \mathit{maj}(a)=\floor*{\frac{\sum_{i=0}^2 a[i]}{2}}
\end{equation*}
to be the majority of $a[0],a[1],a[2]$.

In the canonical algebra $A$, the operations are interpreted as
\begin{align*}
    \mathit{enc}^{A}&: x\mapsto x,x,x\\
    \mathit{flip}^{A}&: (a,i)\mapsto a'\\
    \mathit{corr}^{A}&: a\mapsto \mathit{maj}(a),\mathit{maj}(a),\mathit{maj}(a),
\end{align*}
where $a'[j]=a[j]$ if $j\neq i$ and $a'[i]=a[i]\oplus 1$.
Intuitively, $\mathit{enc}(x)$ is the code word that encodes $x$ by repeating it three times, $\mathit{flip}(a,i)$ flips the physical bit $a[i]$ in the array $a$, and $\mathit{corr}(a)$ produces a code word by replacing every entry of $a$ with $\mathit{maj}(a)$. Then, the algebra $A$ satisfies the following equations:

\begin{equation}
    \label{eq:classical-repetition-correction}
    \begin{aligned}
    \forall x:\texttt{lbit}.\quad
        &\mathit{corr}(\mathit{enc}(x))=\mathit{enc}(x),\\
    \forall x:\texttt{lbit},\,i:\texttt{idx}.\quad
        &\mathit{corr}(\mathit{flip}(\mathit{enc}(x),i))=\mathit{enc}(x).
    \end{aligned}
\end{equation}
The first equation states that $\mathit{corr}$ leaves a code word unchanged.
The second equation states that it recovers the code word after a bit flip on any physical bit.

\subsection{Quantum Repetition Code}
\label{sub:quantum-repetition-code}

Now let us show  the $3$-qubit repetition code as a quantisation of the classical $3$-bit repetition code.
Suppose $\mathit{QSig}=(S_q,\calK,\calU)$ is the quantum signature obtained from a quantisation $\eta$ of the classical signature $\mathit{Sig}$ in \Cref{sub:classical-repetition-code}.
Let
\begin{equation*}
    \texttt{lqbit}=\eta(\texttt{lbit}),\qquad
    \texttt{pqbit}=\eta(\texttt{pbit}),\qquad
    \texttt{qidx}=\eta(\texttt{idx}),\qquad
    \texttt{qarray}=\eta(\texttt{array}).
\end{equation*}
In addition, let $\texttt{qsyn}$ be an ancilla sort for syndromes. Then
\begin{equation*}
    S_q=\braces{\texttt{lqbit},\texttt{pqbit},\texttt{qidx},\texttt{qarray},\texttt{qsyn}}.
\end{equation*}

Consider a quantum algebra $\calA$ with
\begin{align*}
    \calH_{\texttt{lqbit}}^{\calA}
    &=\calH_{\texttt{pqbit}}^{\calA}=\Span\braces*{\ket{0},\ket{1}},\qquad
    \calH_{\texttt{qidx}}^{\calA}
    =\Span\braces*{\ket{0},\ket{1},\ket{2}},\\
    \calH_{\texttt{qarray}}^{\calA}
    &=\bigotimes_{i=0}^2\calH_{\texttt{pqbit}}^{\calA},\qquad\qquad\qquad\quad
    \calH_{\texttt{qsyn}}^{\calA}
    =\Span\braces*{\ket{0},\ket{1},\ket{2},\ket{\bot}}.
\end{align*}
Here, $\ket{\bot}$ is the designated $\ket{0_{\texttt{qsyn}}}$, which is a computational-basis state orthogonal to $\ket{0},\ket{1},\ket{2}$.
As in \Cref{sec-app}, we denote $\ket{a}=\ket{a[0],a[1],a[2]}$.
In particular, the designated $\ket{0_{\texttt{qarray}}}=\ket{0,0,0}$.
We choose the designated functions $k_{\texttt{lbit}}$, $k_{\texttt{pbit}}$, and $k_{\texttt{idx}}$ to be identity functions, and define $k_{\texttt{array}}$ as in \Cref{eq:k-array}.

Next, we explain how the operation symbols $\mathit{enc}$, $\mathit{flip}$, and $\mathit{corr}$ are quantised to gate symbols in $\calU$.

\subsubsection*{Quantum encoding}

The operation $\mathit{enc}^{A}$ can be quantised out-of-place as follows:
\begin{enumerate}
    \item
    $U_{\mathit{enc}}$, $\mathit{pre}_{\mathit{enc}}$, and $\mathit{post}_{\mathit{enc}}$ have sort $\texttt{lqbit},\texttt{qarray}$.
    \item
    $\mathit{pre}_{\mathit{enc}}^{\calA} =\mathit{post}_{\mathit{enc}}^{\calA}=\Id$.
    \item
    $U_{\mathit{enc}}^{\calA}\ket{x}\ket{0,0,0} =\ket{0}\ket{x,x,x}$ for any $x\in A_{\texttt{lbit}}$.
\end{enumerate}
We simply use $\mathit{qenc}$ to denote the gate $U_{\mathit{enc}}$ for convenience.

Compared with the classical $\mathit{enc}$ in \Cref{sub:classical-repetition-code}, the gate $\mathit{qenc}$ encodes $\ket{x}$ as $\ket{x,x,x}$ in the computational basis and resets the logical bit to $\ket{0}$. Note that when the logical qubit is initially in superposition, the physical qubits become entangled.

\subsubsection*{Quantum flip}

The operation $\mathit{flip}^{A}$ can be quantised in-place as follows:
\begin{enumerate}
    \item
    $U_{\mathit{flip}}$, $\mathit{pre}_{\mathit{flip}}$, and $\mathit{post}_{\mathit{flip}}$ have sort $\texttt{qarray},\texttt{qidx}$.
    \item
    $\mathit{pre}_{\mathit{flip}}^{\calA} =\mathit{post}_{\mathit{flip}}^{\calA}=\Id$.
    \item
    $U_{\mathit{flip}}^{\calA}\ket{a}\ket{i} =\ket{\mathit{flip}^{A}(a,i)}\ket{i}$ for any $a\in A_{\texttt{array}}$ and $i\in A_{\texttt{idx}}$.
\end{enumerate}
As usual, we simply use $\mathit{qflip}$ to denote the gate $U_{\mathit{flip}}$ for convenience.

Note that $\mathit{qflip}$ implements the classical operation $\mathit{flip}$ in the computational basis. It leaves the index unchanged and applies a Pauli $X$ gate to the physical qubit indexed by $i$.

\subsubsection*{Quantum correction}

The operation $\mathit{corr}^{A}$ can be quantised in-place as follows:
\begin{enumerate}
    \item
    $U_{\mathit{corr}}$, $\mathit{pre}_{\mathit{corr}}$, and $\mathit{post}_{\mathit{corr}}$ have sort $\texttt{qarray},\texttt{qsyn}$.
    \item
    $\mathit{pre}_{\mathit{corr}}^{\calA} =\mathit{post}_{\mathit{corr}}^{\calA}=\Id$.
    \item
    $U_{\mathit{corr}}^{\calA}\ket{a}\ket{\bot} =\ket{\mathit{corr}^{A}(a)} \ket{\mathit{syn}(a)}$ for any $a\in A_{\texttt{array}}$, where
    \begin{equation}
        \label{eq:syn-repetition}
        \mathit{syn}(a)=
        \begin{cases}
            \bot, & a=\mathit{corr}^{A}(a),\\
            j, & \textup{$j$ is the unique index such that }a[j]\neq \mathit{maj}(a).
        \end{cases}
    \end{equation}
\end{enumerate}
As usual, we simply use $\mathit{qcorr}$ to denote the gate $U_{\mathit{corr}}$ for convenience.

Intuitively, $\mathit{qcorr}$ implements the classical $\mathit{corr}$ in the computational basis and produces the syndrome $\mathit{syn}(a)$.
In the computational basis, when $a$ is a code word, the syndrome is $\bot$; otherwise, it is the index of the unique physical bit that differs from $\mathit{maj}(a)$, which is the error location when $a$ is obtained from a bit-flip on a code word.
Note that $\mathit{qcorr}$ coherently performs the syndrome extraction and error correction together.

According to \Cref{thm:quantise}, the equations in \Cref{eq:classical-repetition-correction} can be lifted to a specification of the quantum repetition code. Here, we give a direct specification that the quantum repetition code corrects a single Pauli $X$ error.

\begin{proposition}
    \label{prop:qec-repetition}
    The quantum algebra $\calA$ with gates $\mathit{qenc}$, $\mathit{qflip}$, and $\mathit{qcorr}$ satisfies the following specification:
    \begin{align}
    \label{eq:qec-repetition-spec-1}
        \begin{split}
        &\forall q_l:\texttt{lqbit}.\quad \Big( \mathit{qenc}[q_l,q_a]; \mathit{qcorr}[q_a,q_s] \ket{0}_{q_a}\ket{\bot}_{q_s}
        = \mathit{qenc}[q_l,q_a] \ket{0}_{q_a}\ket{\bot}_{q_s} \Big)\Downarrow q_a,
        \end{split}
        \\
    \label{eq:qec-repetition-spec-2}
        \begin{split}
        &\forall q_l:\texttt{lqbit},\,q_i:\texttt{qidx}.\quad
        \Big( \mathit{qenc}[q_l,q_a]; \mathit{qflip}[q_a,q_i]; \mathit{qcorr}[q_a,q_s] \ket{0}_{q_a}\ket{\bot}_{q_s} \\[-1mm]
        &\qquad\qquad\qquad\qquad\qquad\qquad = \mathit{qenc}[q_l,q_a] \ket{0}_{q_a}\ket{\bot}_{q_s} \Big)\Downarrow q_a.
        \end{split}
    \end{align}
\end{proposition}
The first equation states that $\mathit{qcorr}$ leaves a quantum code word unchanged. The second equation states that it recovers the quantum code word after $\mathit{qflip}$ introduces a Pauli $X$ error on any physical qubit. In both equations, $q_a$ is the only principal variable.

\subsection{Shor Code}

The $3$-qubit quantum repetition code in \Cref{sub:quantum-repetition-code} can only correct a Pauli $X$ error.
Inspired by the repetition code, the Shor code~\cite{ShorQEC} is the first quantum error-correcting code that can correct an arbitrary single-qubit error.
Let us examine the $9$-qubit Shor code as a quantisation of a $9$-bit repetition code.
We reuse the classical signature $\mathit{Sig}$ in \Cref{sub:classical-repetition-code} with slight modifications:
the operation symbol $\mathit{flip}$ in \Cref{sub:classical-repetition-code} is replaced by two operation symbols $\mathit{flipX},\mathit{flipZ}:\texttt{array},\texttt{idx}\rightarrow\texttt{array}$ of the same sort.

Consider the canonical algebra $B$ for a $9$-bit repetition code. It gives
\begin{equation*}
    B_{\texttt{lbit}}=B_{\texttt{pbit}}=\braces*{0,1},
    \qquad B_{\texttt{idx}}=\braces*{0,1,\ldots,8},
    \qquad B_{\texttt{array}}=B_{\texttt{idx}}\rightarrow B_{\texttt{pbit}}.
\end{equation*}
For $a\in B_{\texttt{array}}$ and $r\in\braces*{0,1,2}$, we define
\begin{equation*}
    m(a)[r] = \floor*{ \frac{\sum_{t=0}^{2}a[3r+t]}{2} },
    \qquad    \mathit{maj}(a) = \floor*{ \frac{\sum_{r=0}^{2}m(a)[r]}{2} }.
\end{equation*}
Intuitively, the $9$ physical bits of $a$ are split into three groups,
where the $r^{\textup{th}}$ group consists of $a[3r],a[3r+1],a[3r+2]$.
The value $m(a)[r]$ is the majority of the $r^{\textup{th}}$ group, while $\mathit{maj}(a)$ is the majority of $m(a)[0],m(a)[1],m(a)[2]$.

In the algebra $B$, the operations are interpreted as 
\begin{align*}
    \mathit{enc}^{B}(x)[i]&=x,\\
    \mathit{flipX}^{B}(a,i)[j] &= \mathit{flipZ}^{B}(a,i)[j] =
    \begin{cases}
        a[j]\oplus 1, &j=i,\\
        a[j], &j\neq i,
    \end{cases},\\
    \mathit{corr}^{B}(a)[i]&=\mathit{maj}(a),
\end{align*}
where $x\in B_{\texttt{lbit}}$, $a\in B_{\texttt{array}}$, and $i,j\in B_{\texttt{idx}}$.
Here, although both $\mathit{flipX}$ and $\mathit{flipZ}$ are interpreted as a bit flip error in algebra $B$,
their quantisations will correspond to different quantum errors later.

Similar to the algebra $A$ for the $3$-bit repetition code (see \Cref{sub:classical-repetition-code}), the algebra $B$ satisfies
\begin{equation}
    \label{eq:classical-shor-correction}
    \begin{aligned}
        \forall x:\texttt{lbit}.\quad
        &\mathit{corr}(\mathit{enc}(x))=\mathit{enc}(x),\\
        \forall x:\texttt{lbit},\,i:\texttt{idx}.\quad
        &\mathit{corr}(\mathit{flipX}(\mathit{enc}(x),i))
        =\mathit{enc}(x),\\
        \forall x:\texttt{lbit},\,i:\texttt{idx}.\quad
        &\mathit{corr}(\mathit{flipZ}(\mathit{enc}(x),i))
        =\mathit{enc}(x).
    \end{aligned}
\end{equation}

Suppose $\mathit{QSig}=(S_q,\calK,\calU)$ is the quantum signature
obtained from a quantisation $\eta$ of this signature. Let
\begin{equation*}
    \texttt{lqbit}=\eta(\texttt{lbit}),\qquad
    \texttt{pqbit}=\eta(\texttt{pbit}),\qquad
    \texttt{qidx}=\eta(\texttt{idx}),\qquad
    \texttt{qarray}=\eta(\texttt{array}).
\end{equation*}
In addition, let $\texttt{qsyn}$ be an ancilla sort for syndromes, with designated state $\ket{0_{\texttt{qsyn}}}=\ket{\bot,\bot,\bot,\bot}$.
Consider a quantum algebra $\calB$ with
\begin{align*}
    \calH_{\texttt{lqbit}}^{\calB}
    &=\calH_{\texttt{pqbit}}^{\calB}
    =\Span\braces*{\ket{0},\ket{1}},\qquad
    \calH_{\texttt{qidx}}^{\calB}
    =\Span\braces*{\ket{0},\ket{1},\ldots,\ket{8}},\\
    \calH_{\texttt{qarray}}^{\calB}
    &=\bigotimes_{i=0}^{8}\calH_{\texttt{pqbit}}^{\calB},\qquad\qquad
    \calH_{\texttt{qsyn}}^{\calB}
    =\Span\braces*{\ket{s}:s\in \braces{0,1,2,\bot}^4}.
\end{align*}
As usual, we denote $\ket{a}=\ket{a[0],\ldots,a[8]}$.
In particular, the designated state $\ket{0_{\texttt{qarray}}}=\ket{0}^{\otimes 9}$.
For a single-qubit unitary operator $U$, let $U_i=\Id^{\otimes i}\otimes U\otimes \Id^{\otimes (8-i)}$, which applies $U$ on the qubit indexed by $i$.
The designated functions $k_{\texttt{lbit}}$, $k_{\texttt{pbit}}$, $k_{\texttt{idx}}$, and $k_{\texttt{array}}$ are chosen as in \Cref{sub:quantum-repetition-code}.

Next, we explain how the operation symbols $\mathit{enc}$, $\mathit{flipX}$, $\mathit{flipZ}$, and $\mathit{corr}$ are quantised to gate symbols in $\calU$.

\subsubsection*{Quantum encoding}

For $x\in\braces*{0,1}$, let
\begin{equation}
    \label{eq:shor-code-word}
    \ket{\Psi_x} = \parens*{ \frac{\ket{000}+(-1)^x\ket{111}}{\sqrt{2}} }^{\otimes3}.
\end{equation}

The operation $\mathit{enc}^{B}$ can be quantised out-of-place as follows:
\begin{enumerate}
    \item
    $U_{\mathit{enc}}$, $\mathit{pre}_{\mathit{enc}}$, and
    $\mathit{post}_{\mathit{enc}}$ have sort
    $\texttt{lqbit},\texttt{qarray}$.
    \item
    $\mathit{pre}_{\mathit{enc}}^{\calB}=\Id$, and $\mathit{post}_{\mathit{enc}}^{\calB} \ket{0}\ket{\Psi_x} = \ket{0}\ket{\mathit{enc}^{B}(x)}$ for any $x\in B_{\texttt{lbit}}$.
    \item
    $U_{\mathit{enc}}^{\calB} \ket{x}\ket{0}^{\otimes9} = \ket{0}\ket{\Psi_x}$ for any $x\in B_{\texttt{lbit}}$.
\end{enumerate}
We simply use $\mathit{qenc}$ to denote the gate $U_{\mathit{enc}}$ for convenience.

Compared with \Cref{sub:quantum-repetition-code}, the gate $\mathit{qenc}$ here encodes the computational-basis $\ket{x}$ into a Shor code word $\ket{\Psi_x}$. If the logical qubit is initially in superposition, the physical qubits are in the corresponding superposition of $\ket{\Psi_0}$ and $\ket{\Psi_1}$.

\subsubsection*{Quantum flip}

The operations $\mathit{flipX}^{B}$ and $\mathit{flipZ}^B$ can be quantised in-place as follows:
\begin{enumerate}
    \item
    $U_{\mathit{flipX}}$, $U_{\mathit{flipZ}}$, $\mathit{pre}_{\mathit{flipX}}$, $\mathit{pre}_{\mathit{flipZ}}$, $\mathit{post}_{\mathit{flipX}}$, and $\mathit{post}_{\mathit{flipZ}}$ have sort
    $\texttt{qarray},\texttt{qidx}$.
    \item
    $\mathit{pre}_{\mathit{flipX}}^{\calB}
    =\mathit{post}_{\mathit{flipX}}^{\calB}=\Id$ and 
    $\mathit{pre}_{\mathit{flipZ}}^{\calB}
    =\mathit{post}_{\mathit{flipZ}}^{\calB}=(\prod_{i=0}^8 H_i)\otimes \Id$,
    \item
    $U_{\mathit{flipX}}^{\calB} =\sum_{i=0}^8 X_i\otimes \ket{i}\!\bra{i}$ and 
    $U_{\mathit{flipZ}}^{\calB} =\sum_{i=0}^8 Z_i\otimes \ket{i}\!\bra{i}$.
\end{enumerate}
We simply use $\mathit{qflipX}$ and $\mathit{qflipZ}$ to denote the gates $U_{\mathit{flipX}}$ and $U_{\mathit{flipZ}}$, respectively, for convenience.

The gates $\mathit{qflipX}$ and $\mathit{qflipZ}$ leave the index $\ket{i}$ unchanged and apply Pauli $X$ and $Z$ gates, respectively, on the physical qubit indexed by $i$. 
Thus, they introduce the corresponding Pauli errors.

\subsubsection*{Quantum correction}

For $a\in B_{\texttt{array}}$, we define
$\mathit{syn}(a)\in\braces*{0,1,2,\bot}^{4}$ as follows.
For $r\in\braces*{0,1,2}$,
\begin{equation*}
    \mathit{syn}(a)[r]
    =
    \begin{cases}
        \bot, &a[3r+t]=m(a)[r] \text{ for every }t\in\braces*{0,1,2},\\
        t, &\textup{$t$ is the unique index such that }a[3r+t]\neq m(a)[r].
    \end{cases}
\end{equation*}
Further, we define
\begin{equation*}
    \mathit{syn}(a)[3] =
    \begin{cases}
        \bot, &m(a)[r]=\mathit{maj}(a) \text{ for every }r\in\braces*{0,1,2},\\
        r, &\textup{$r$ is the unique index such that }m(a)[r]\neq\mathit{maj}(a).
    \end{cases}
\end{equation*}
Intuitively, $\mathit{syn}(a)$ is the syndrome produced for an arbitrary $a$. 
For each $r\in \braces{0,1,2}$, if there is an entry in the $r^{\textup{th}}$ group of $a$ that differs from $m(a)[r]$, then $\mathit{syn}(a)[r]$ records the index of this entry; otherwise, it records $\bot$. Similarly, if an entry $m(a)[r]$ differs from $\mathit{maj}(a)$, then $\mathit{syn}(a)[3]$ records the index $r$; otherwise, it records $\bot$. 
Consequently, it is easy to check that $a\mapsto(\mathit{maj}(a),\mathit{syn}(a))$ is a bijection.

For $P\in \braces{\Id, X,Y,Z}$, $i=3r+t$, and $k\in \braces{0,1,2,3}$, where $r,t\in\braces*{0,1,2}$, let us define $\sigma(P,i)\in\braces*{0,1,2,\bot}^{4}$ by
\begin{equation}
    \label{eq:sigma}
    \sigma(P,i)[k] =
    \begin{cases}
        t, &k=r\text{ and }P\in\braces*{X,Y},\\
        r, &k=3\text{ and }P\in\braces*{Z,Y},\\
        \bot, &o.w.
    \end{cases}
\end{equation}
Intuitively, $\sigma(P,i)$ is the syndrome needed to correct the Pauli error $P_i$. 
Suppose we start with a Shor code word.
A Pauli $X$ error changes one physical qubit, so its syndrome records the information about $i=3r+t$.
A Pauli $Z$ error on any of the three physical qubits in the same group $r$ has the same effect, so its syndrome only records the $r$. A Pauli $Y$ error combines $X$ and $Z$, so its syndrome contains both kinds of information.

Let $\braces{\ket{\Phi_{x,s}}:x\in \braces{0,1},s\in \braces*{0,1,2,\bot}^4}$ be an orthonormal basis of $\calH_{\texttt{qarray}}^{\calB}$ such that
\begin{equation}
    \label{eq:def-Phi}
    \ket{\Phi_{x,\sigma(P,i)}} = P_i\ket{\Psi_x}
\end{equation}
for $x\in B_{\texttt{lbit}}$, $P\in \braces{\Id, X,Y,Z}$, and $i\in B_{\texttt{idx}}$.
The existence of such a basis is proved in \Cref{lem:exist-Phi}.
Here, in $\ket{\Phi_{x,s}}$, $x$ records the logical bit and $s$ records the syndrome. 

The operation $\mathit{corr}^{B}$ can be quantised in-place as follows:
\begin{enumerate}
    \item
    $U_{\mathit{corr}}$, $\mathit{pre}_{\mathit{corr}}$, and
    $\mathit{post}_{\mathit{corr}}$ have sort
    $\texttt{qarray},\texttt{qsyn}$.

    \item
    $\mathit{pre}_{\mathit{corr}}^{\calB} \ket{a}\ket{\bot,\bot,\bot,\bot} = \ket{\Phi_{\mathit{maj}(a),\mathit{syn}(a)}} \ket{\bot,\bot,\bot,\bot}$ and
    $\mathit{post}_{\mathit{corr}}^{\calB} \ket{\Psi_x}\ket{s} = \ket{\mathit{enc}^{B}(x)}\ket{s}$ for every $x\in B_{\texttt{lbit}}$, $a\in B_{\texttt{array}}$, and $s\in \braces{0,1,2,\bot}^4$.

    \item
    $ U_{\mathit{corr}}^{\calB} \ket{\Phi_{x,s}} \ket{\bot,\bot,\bot,\bot} = \ket{\Psi_x}\ket{s}$ for every $x\in B_{\texttt{lbit}}$ and $s\in\braces*{0,1,2,\bot}^{4}$.
\end{enumerate}
We simply use $\mathit{qcorr}$ to denote the gate $U_{\mathit{corr}}$ for convenience.

Intuitively, $\mathit{pre}_{\mathit{corr}}$ implements the bijection $a \mapsto (\mathit{maj}(a),\mathit{syn}(a))$ by mapping $\ket{a}$ to $\ket{\Phi_{\mathit{maj}(a),\mathit{syn}(a)}}$, while leaving the initial syndrome in $q_s$ unchanged.
The quantum gate $\mathit{qcorr}$ maps $\ket{\Phi_{x,s}}$ to the corrected Shor code word $\ket{\Psi_x}$ and records the syndrome in $q_s$.
Finally, $\mathit{post}_{\mathit{corr}}$ maps $\ket{\Psi_x}$ to the computational-basis repetition code word $\ket{\mathit{enc}^B(x)}$.

As in \Cref{sec-app}, we can derive specifications for the Shor code by lifting \Cref{eq:classical-repetition-correction} using \Cref{thm:quantise}.
Beyond that, the Shor code actually satisfies stronger error-correction properties.
The following equations specify that the Shor code corrects a single Pauli $X$, $Z$, or $Y$ error.

\begin{proposition}
    \label{prop:qec-shor}
    The quantum algebra $\calB$ with gates $\mathit{qenc}$,
    $\mathit{qflipX}$, $\mathit{qflipZ}$, and $\mathit{qcorr}$
    satisfies the following specification:
    \begin{align}
    \label{eq:qec-shor-spec-1}
        \begin{split}
        &\forall q_l:\texttt{lqbit}.\quad \Big( \mathit{qenc}[q_l,q_a]; \mathit{qcorr}[q_a,q_s] \ket{0}_{q_a}\ket{\bot,\bot,\bot,\bot}_{q_s} \\[-1mm]
        &\qquad\qquad\qquad = \mathit{qenc}[q_l,q_a] \ket{0}_{q_a}\ket{\bot,\bot,\bot,\bot}_{q_s} \Big)\Downarrow q_a,
        \end{split}
        \\
    \label{eq:qec-shor-spec-2}
        \begin{split}
        &\forall q_l:\texttt{lqbit},\,q_i:\texttt{qidx}.\quad \Big( \mathit{qenc}[q_l,q_a]; \mathit{qflipX}[q_a,q_i]; \mathit{qcorr}[q_a,q_s]\ket{0}_{q_a}\ket{\bot,\bot,\bot,\bot}_{q_s}\\[-1mm]
        &\qquad\qquad\qquad\qquad\qquad = \mathit{qenc}[q_l,q_a] \ket{0}_{q_a}\ket{\bot,\bot,\bot,\bot}_{q_s} \Big)\Downarrow q_a,
        \end{split}
        \\
    \label{eq:qec-shor-spec-3}
        \begin{split}
        &\forall q_l:\texttt{lqbit},\,q_i:\texttt{qidx}.\quad \Big( \mathit{qenc}[q_l,q_a]; \mathit{qflipZ}[q_a,q_i]; \mathit{qcorr}[q_a,q_s] \ket{0}_{q_a}\ket{\bot,\bot,\bot,\bot}_{q_s}\\[-1mm]
        &\qquad\qquad\qquad\qquad\qquad = \mathit{qenc}[q_l,q_a] \ket{0}_{q_a}\ket{\bot,\bot,\bot,\bot}_{q_s} \Big)\Downarrow q_a,
        \end{split}
        \\
    \label{eq:qec-shor-spec-4}
        \begin{split}
        &\forall q_l:\texttt{lqbit},\,q_i:\texttt{qidx}.\quad \Big( \mathit{qenc}[q_l,q_a]; \mathit{qflipX}[q_a,q_i]; \mathit{qflipZ}[q_a,q_i]; \mathit{qcorr}[q_a,q_s] \ket{0}_{q_a}\ket{\bot,\bot,\bot,\bot}_{q_s}\\[-1mm]
        &\qquad\qquad\qquad\qquad\qquad = \mathit{qenc}[q_l,q_a] \ket{0}_{q_a}\ket{\bot,\bot,\bot,\bot}_{q_s} \Big)\Downarrow q_a.
        \end{split}
    \end{align}
\end{proposition}
The first equation states that $\mathit{qcorr}$ leaves an error-free Shor code word unchanged. The second and third equations state that $\mathit{qcorr}$ recovers the code word after a Pauli $X$ or $Z$ error, respectively. 
The fourth equation covers a Pauli $Y$ error through sequential application of $\mathit{qflipX}$ and $\mathit{qflipZ}$. In all equations, $q_a$ is the only principal variable.
Following the standard quantum error-correction argument, being able to correct Pauli $X,Y,Z$ errors means that the Shor code corrects an arbitrary single-qubit error.

\section{Embeddings of Quantum Data Types}\label{sec-emb}

So far, we have only considered the quantisation of individual ADTs. In classical computing, various relationships between different ADTs have been exploited for practical applications in software engineering. In particular, an embedding allows values of one ADT to be represented as values of another while preserving the relevant operations and semantics. Consequently, embeddings between ADTs can be used for applications such as data refinement and software component reuse. To enable these kinds of applications for quantum ADTs, in this section we ask under what conditions quantisation preserves structural relationships between algebras. In classical universal algebra, an embedding faithfully maps one algebra into another while preserving all operations. This section investigates when such an embedding lifts to an embedding between the corresponding quantum algebras.

\begin{definition}[Homomorphisms of algebras] 
    \label{def:homo}
    Let $A$ and $A^\prime$ be two algebras of signature $\mathit{Sig}=(S,\mathit{OP})$. Then: \begin{enumerate}\item A homomorphism from $A$ to $A^\prime$ is a family $g=\parens*{g_s}_{s\in S}$ of mappings $g_s:A_s\rightarrow A^\prime_s$ such that for each operation symbol $O:s_1,\ldots,s_n\rightarrow s$, and for all $a_1\in A_{s_1},\ldots,a_n\in A_{s_n},$ it holds that $$g_s(O^A(a_1,\ldots,a_n))=O^{A^\prime}(g_{s_1}(a_1),\ldots,g_{s_n}(a_n)).$$ 
Equivalently, the following diagram commutes:
\begin{equation*}
\begin{tikzcd}[
    column sep=3.8em,
    row sep=3.3em
]
A_{s_1}\times\cdots\times A_{s_n}
    \arrow[r, "O^A"]
    \arrow[d, "g_{s_1}\times\cdots\times g_{s_n}" swap]
& A_s \arrow[d, "g_s"]
\\\
A_{s_1}^\prime\times\cdots\times A_{s_n}^\prime
    \arrow[r, "O^{A^\prime}" swap]
& A_s^\prime
\end{tikzcd}
\end{equation*}
\item If all $g_s$ are injective, then $g$ is called an embedding. 
\item If all $g_s$ are bijective, then $g$ is called an isomorphism. \end{enumerate} 
\end{definition}

Let $\calH$ and $\calH^\prime$ be two Hilbert spaces. A linear operator $T:\calH\rightarrow\calH^\prime$ is called an isometric embedding if $$\braket{T(x)|T(y)}_{\calH^\prime}=\braket{x|y}_\calH$$ for all $x,y\in\calH$, where $\braket{\cdot|\cdot}_\calH, \braket{\cdot|\cdot}_{\calH^\prime}$ denote the inner products in $\calH$ and $\calH^\prime$, respectively. It is easy to see that $T$ is injective and norm-preserving: $$\norm{T(x)}_{\calH^\prime}=\norm{x}_\calH$$ for every $x\in\calH$, and the image $T(\calH)$ is a closed subspace of $\calH^\prime$. 

\begin{definition}[Embeddings between quantum algebras]\label{def-qemb}Let $\calA$ and $\calA^\prime$ be two quantum algebras of signature $\mathit{Sig}_q=(S_q,\calK,\calU)$. Then an embedding from $\calA$ to $\calA^\prime$ is a family $G=\parens*{G_s}_{s\in S_q}$ of isometric embeddings $G_s:\calH_s\rightarrow \calH^\prime_s$ such that \begin{enumerate}\item for each quantum state symbol $\ket{\psi}\in\calK$ of sort $s_1,\ldots,s_n$, it holds that 
    \begin{equation}
        \label{eq:emb-cond-1}
        (G_{s_1}\otimes\cdots\otimes G_{s_n})(\ket{\psi}^\calA)=\ket{\psi}^{\calA^\prime};
    \end{equation}
\item for each quantum gate symbol $U\in\calU$ of sort $s_1,\ldots,s_n$, and for any $\ket{\varphi}\in\bigotimes_{i=1}^n\calH_{s_i}$, 
\begin{equation}\label{eq-hom}(G_{s_1}\otimes\cdots\otimes G_{s_n})(U^\calA\ket{\varphi})=U^{\calA^\prime}((G_{s_1}\otimes\cdots\otimes G_{s_n})\ket{\varphi}).\end{equation}  
Equivalently, the following diagram commutes:
\begin{equation*}
\begin{tikzcd}[
    column sep=3.8em,
    row sep=3.3em
]
\displaystyle \bigotimes_{i=1}^n\calH_{s_i}
    \arrow[r, "U^\calA"]
    \arrow[d, "G_{s_1}\otimes\cdots\otimes G_{s_n}" swap]
& \displaystyle \bigotimes_{i=1}^n\calH_{s_i}
    \arrow[d, "G_{s_1}\otimes\cdots\otimes G_{s_n}"]
\\
\displaystyle \bigotimes_{i=1}^n\calH_{s_i}^\prime
    \arrow[r, "U^{\calA^\prime}" swap]
& \displaystyle \bigotimes_{i=1}^n\calH_{s_i}^\prime
\end{tikzcd}
\end{equation*}
\end{enumerate}
\end{definition}

Since $G_{s_1},\ldots,G_{s_n}$ are all isometric embeddings, $\calH^\prime\stackrel{\triangle}{=}(G_{s_1}\otimes\cdots\otimes G_{s_n})\parens*{\bigotimes_{i=1}^n\calH_{s_i}}$ is a closed subspace of $\bigotimes_{i=1}^n\calH_{s_i}^\prime$. Moreover, $G_{s_1}\otimes\cdots\otimes G_{s_n}:\bigotimes_{i=1}^n\calH_{s_i}\rightarrow \calH^\prime$ and its inverse $(G_{s_1}\otimes\cdots\otimes G_{s_n})^\dag:\calH'\rightarrow\bigotimes_{i=1}^n\calH_{s_i} $ are unitary operators. Thus, the condition in \Cref{eq-hom} can be rewritten as:
$$U^\calA=(G_{s_1}\otimes\cdots\otimes G_{s_n})^\dag\circ U^{\calA^\prime}\circ (G_{s_1}\otimes\cdots\otimes G_{s_n})$$. If all $G_s$ are unitary operators, then $G$ is called an isomorphism. In this case, $G_s(\calH_s)=\calH_s^\prime$ for every $s\in S_q$, and $\calH^\prime=\bigotimes_{i=1}^n\calH_{s_i}^\prime$. 

In the remainder of this section, we give conditions under which an embedding of algebras induces an embedding of their quantisations.

\subsection{Embeddings of Quantisations} 

Let $A$ and $A'$ be two algebras of the same signature
$\mathit{Sig}=(S,\mathit{OP})$,
and let $g=(g_s)_{s\in S}:A\rightarrow A^\prime$ be an embedding.
Let $\eta$ be a quantisation from $\mathit{Sig}$ to
$\mathit{QSig}=(S_q,\calK,\calU)$ that satisfies $S_q=\eta(S)$.
Let $\calA$ and $\calA^\prime$ be quantisations of $A$ and $A^\prime$ along $\eta$, respectively.

Recall that the quantisation of $s\in S$ is $s_q\in S_q$ along $\eta$.
For each $s\in S$, we define $G_{s_q}:\calH_{s_q}^{\calA}\rightarrow \calH_{s_q}^{\calA'}$ to be the linear operator that satisfies the following condition:
\begin{equation}
    \label{eq-emb}
    G_{s_q}(\ket{a})=\ket{g_{s}(a)}\ \mbox{for every}\ a\in A_{s}.
\end{equation} 
We call $G=(G_{s_q})_{s\in S}$ the quantisation of $g=(g_s)_{s\in S}$ along $\eta$.  
Since each $g_s$ is injective, $\braces{\ket{g_s(a)}:a\in A_s}$ is an orthonormal set,
which means each $G_{s_q}$ is an isometric embedding.
If $g$ is further an isomorphism, then each $G_{s_q}$ is a unitary.

For concreteness, we consider BO- and PO-quantisations in \Cref{sec-quantisation}.
For each $s\in S$, let $k_s^A:A_s\rightarrow \Z_{\abs*{A_s}}$ and $k_s^{A'}:A_s'\rightarrow \Z_{\abs*{A_s'}}$ be the two designated bijections as in \Cref{sec-quantisation},
and let $\boxplus_s^A$ and $\boxplus_s^{A'}$ be defined as in \Cref{sec-quantisation}.

We first examine when a BO-quantisation can be embedded into another.

\begin{proposition}[Embeddings of BO-quantisations]
    \label{prop:BO-embedding}
    Suppose that every state symbol in $\calK$ is introduced by $\eta$, and that every gate symbol in $\calU$ is introduced by $\eta$.
    Let $\calA$ and $\calA'$ be BO-quantisations of $A$ and $A'$, respectively, along $\eta$.
    If for each $s\in S$ and $a,b\in A_s$, 
    \begin{equation*}
        g_s(a\boxplus_s^A b) =g_s(a)\boxplus_s^{A'} g_s(b),
    \end{equation*}
    or equivalently, the following diagram commutes:
    \begin{equation*}
    \begin{tikzcd}[column sep=large, row sep=large]
        A_s\times A_s \arrow[r, "\boxplus_s^A"]
            \arrow[d, "g_s\times g_s"']
        & A_s \arrow[d, "g_s"] \\
        A_s'\times A_s' \arrow[r, "\boxplus_s^{A'}"']
        & A_s'
    \end{tikzcd}
    \end{equation*}
    then $G$ in \Cref{eq-emb} is an embedding from $\calA$ to $\calA'$.
\end{proposition}

Next, we examine when a PO-quantisation is isomorphic to another.
Unlike the BO-quantisations, compatibility with the quantum Fourier transform in \Cref{def:po-quantisation}
forces every $g_s$ to be surjective.
Consequently, the corresponding result for PO-quantisations is stated only for isomorphisms.

\begin{proposition}[Isomorphisms of PO-quantisations]
    \label{prop:PO-embedding}
    Suppose that every state symbol in $\calK$ is introduced by $\eta$, and that every gate symbol in $\calU$ is introduced by $\eta$.
    Suppose $g$ is an isomorphism.
    Let $\calA$ and $\calA'$ be PO-quantisations of $A$ and $A'$, respectively, along $\eta$.
    If for each $s\in S$,
    \begin{equation*}
        k_s^{A'}\circ g_s=k_s^A,
    \end{equation*}
    or equivalently, the following diagram commutes (noting that $\abs{A_s}=\abs{A_s'}$):
    \begin{equation*}
    \begin{tikzcd}[column sep=large, row sep=large]
        A_s \arrow[r, "k_s^A"] \arrow[d, "g_s"']
        & \Z_{\abs*{A_s}} \arrow[d, equal] \\
        A_s' \arrow[r, "k_s^{A'}"']
        & \Z_{\abs*{A_s'}}
    \end{tikzcd}
    \end{equation*}
    then $G$ in \Cref{eq-emb} is an isomorphism from $\calA$ to $\calA'$.
\end{proposition}

\section{Products of Quantum Data Types}\label{sec-prod}

Embeddings discussed in the last section describe one of the simplest structural relationships between quantum ADTs. We now turn to how to construct a new quantum ADT from simpler ones. In classical computing, products of ADTs enable us to compose independently specified abstractions while preserving their interfaces and invariants, and they find applications in composite data structures, modular software composition, and other settings. To introduce the same capability into quantum software engineering, in this section we define the notion of the product of ADTs. At the classical level, the domains of two algebras are combined using Cartesian products. At the quantum level, the corresponding quantum domains are combined using tensor products. Furthermore, we present the conditions under which these products are compatible with the quantisations of ADTs.

\begin{definition}Let $A$ and $B$ be two algebras of signature $\mathit{Sig}=(S,\mathit{OP})$. Then their product is defined as the algebra $C=A\times B$ of the same signature $\mathit{Sig}$ satisfying the following conditions:
\begin{enumerate}\item for each sort $s\in S$, $C_s=A_s\times B_s$ (Cartesian product);
\item for each operation symbol $O:s_1,\ldots,s_n\rightarrow s$, $O^C$ is defined as follows: for any $a_i\in A_{s_i}$ and $b_i\in B_{s_i}$ $(i=1,\ldots,n)$,
\begin{equation}\label{c-prod}O^C((a_1,b_1),\ldots,(a_n,b_n))=(O^A(a_1,\ldots,a_n),O^B(b_1,\ldots,b_n)).\end{equation} In particular, for each constant symbol $c:\rightarrow s$, we have $c^C=(c^A,c^B)$. 
\end{enumerate}
\end{definition}

We next define the corresponding product of quantum algebras. Let $\overline{s}=s_1,\ldots,s_n$ be a string of quantum sorts. Define the canonical unitary operator
\begin{equation*}
    T_{\overline{s}}:\bigotimes_{i=1}^n
    \parens*{\calH_{s_i}^{\calA}\otimes\calH_{s_i}^{\calB}}
    \rightarrow \parens*{\bigotimes_{i=1}^n\calH_{s_i}^{\calA}} \otimes \parens*{\bigotimes_{i=1}^n\calH_{s_i}^{\calB}}
\end{equation*}
by 
\begin{equation}
    \label{U-product}
    T_{\overline{s}}\parens*{\bigotimes_{i=1}^n\ket{\alpha_i\beta_i}}=\parens*{\bigotimes_{i=1}^n\ket{\alpha_i}}\otimes \parens*{\bigotimes_{i=1}^n\ket{\beta_i}}
\end{equation}
for any $\ket{\alpha_i}\in\calH_{s_i}^\calA$ and $\ket{\beta_i}\in\calH_{s_i}^\calB$ $(i=1,\ldots,n)$
together with linear extension.

\begin{definition}Let $\calA$ and $\calB$ be two quantum algebras of signature $\mathit{QSig}=(S_q, \calK,\calU)$. Then their product is defined as the quantum algebra $\calC=\calA\otimes\calB$ of the same signature $\mathit{QSig}$ satisfying the following conditions:\begin{enumerate}\item for each sort $s\in S_q$, $\calH^\calC_s=\calH_s^\calA\otimes\calH_s^\calB$ (tensor product);
\item for each quantum state symbol $\ket{\psi}\in\calK$ of sort $s_1,\ldots,s_n$, $\ket{\psi}^\calC=T^{-1}_{s_1,\ldots,s_n}\parens{\ket{\psi}^\calA\otimes \ket{\psi}^\calB}$;
\item for each quantum gate symbol $U\in\calU$ of sort $s_1,\ldots,s_n$, \begin{equation}\label{u-p}U^\calC=T_{s_1,\ldots,s_n}^{-1}\circ (U^\calA\otimes U^\calB)\circ T_{s_1,\ldots,s_n}\end{equation} is a unitary operator on $\bigotimes_{i=1}^n\calH^\calC_{s_i}=\bigotimes_{i=1}^n\parens*{\calH_{s_i}^\calA\otimes\calH_{s_i}^\calB}$, where $U^\calA\otimes U^\calB$ is defined by
\begin{align}
    \parens*{U^\calA\otimes U^\calB}(\ket{\alpha_1\ldots\alpha_n}\otimes\ket{\beta_1\ldots\beta_n})&=U^\calA(\ket{\alpha_1\ldots\alpha_n})\otimes U^\calB(\ket{\beta_1\ldots\beta_n}),
\end{align} 
for any $\ket{\alpha_i}\in\calH_{s_i}^\calA$ and $\ket{\beta_i}\in\calH_{s_i}^\calB$ $(i=1,\ldots,n)$ and extending linearly.
Equivalently, \Cref{u-p} can be expressed by the following commutative diagram:
\begin{equation*}
\begin{tikzcd}[
    column sep=4.8em,
    row sep=3.3em,
    cells={nodes={inner sep=2pt}}
]
\displaystyle \bigotimes_{i=1}^n\calH^\calC_{s_i}
    \arrow[r, "U^\calC"]
    \arrow[d, "T_{s_1,\ldots,s_n}"']
&
\displaystyle \bigotimes_{i=1}^n\calH^\calC_{s_i}
\\
\displaystyle \parens*{\bigotimes_{i=1}^n\calH^\calA_{s_i}}\otimes \parens*{\bigotimes_{i=1}^n\calH^\calB_{s_i}}
    \arrow[r, "U^\calA\otimes U^\calB"']
&
\displaystyle \parens*{\bigotimes_{i=1}^n\calH^\calA_{s_i}}\otimes\parens*{\bigotimes_{i=1}^n\calH^\calB_{s_i}}
    \arrow[u, "T^{-1}_{s_1,\ldots,s_n}"']
\end{tikzcd}
\end{equation*}
\end{enumerate}
\end{definition}

\subsection{Products of Quantisations}

Let $A$ and $B$ be two algebras of the same signature $\mathit{Sig}=(S,\mathit{OP})$, and let $C=A\times B$.
Suppose $\eta$ is a quantisation from $\mathit{Sig}$ to $\mathit{QSig}$ that satisfies $S_q=\eta(S)$.
Suppose that $\calA$, $\calB$, and $\calC$ are quantisations of $A$, $B$, and $C$, respectively, along $\eta$.

Recall that the quantisation of $s\in S$ is $s_q\in S_q$ along $\eta$.
For each $s\in S$, we define $J_{s_q}: \calH_{s_q}^{\calC} \rightarrow \calH_{s_q}^{\calA}\otimes\calH_{s_q}^{\calB}$
to be a linear operator that satisfies the following condition:
\begin{equation}
    \label{eq:product-isomorphism}
    J_{s_q}\ket{(a,b)} = \ket{a}\otimes\ket{b}
\end{equation}
for every $a\in A_s$ and $b\in B_s$.
Note that $J_{s_q}$ is a unitary. Let $J=(J_{s_q})_{s\in S}$.

For concreteness, we consider when a BO-quantisation (and PO-quantisation) is compatible with products.
For each $D\in \braces{A,B,C}$ and $s\in S$,
let $k_s^D:D_s\rightarrow \Z_{\abs{D_s}}$ be the designated bijection,
and let $\boxplus_s^D$ be defined as in \Cref{sec-quantisation}.

\begin{proposition}[Products of BO-quantisations]
\label{prop:BO-product}
    Suppose that every state symbol in $\calK$ is introduced by $\eta$, and that every gate symbol in $\calU$ is introduced by $\eta$.
Let $\calA$, $\calB$, and $\calC$ be BO-quantisations of $A$, $B$, and $C=A\times B$, respectively, along the same quantisation $\eta$. If
\begin{equation}
    \label{eq:product-sum-compatible}
    (a,b)\boxplus_s^C(a',b') = \parens*{a\boxplus_s^A a', b\boxplus_s^B b' }
\end{equation}
for any $s\in S$, $a,a'\in A_s$, and $b,b'\in B_s$,
then $J$ in \Cref{eq:product-isomorphism} is an isomorphism from $\calC$ to $\calA\otimes \calB$.
\end{proposition}

\begin{proposition}[Products of PO-quantisations]
\label{prop:PO-product}
    Suppose that every state symbol in $\calK$ is introduced by $\eta$, and that every gate symbol in $\calU$ is introduced by $\eta$.
Let $\calA$, $\calB$, and $\calC$ be PO-quantisations of $A$, $B$, and $C=A\times B$, respectively, along the same quantisation $\eta$. 
If
\begin{equation}
\label{eq:product-phase-compatible}
    \exp\parens*{ \frac{2\pi i\,k_s^C((a,b))k_s^C((a',b'))}{|C_s|}}
    = \exp\parens*{ \frac{2\pi i\,k_s^A(a)k_s^A(a')}{|A_s|} }
    \exp\parens*{\frac{2\pi i\,k_s^B(b)k_s^B(b')}{|B_s|}}
\end{equation}
for any $s\in S$, $a,a'\in A_s$, and $b,b'\in B_s$,
then $J$ is an isomorphism from $\calC$ to $\calA\otimes \calB$.
\end{proposition}

\section{Conclusion}

This paper developed a universal-algebraic framework for data abstraction in quantum programming. We formally defined the notion of abstract quantum data types and a quantisation of classical data types and proved that their equational specifications can be soundly lifted to their quantisations. Generalised bit-oracle and phase-oracle quantisations arise as instances of this construction and are related through phase kickback. Further, the examples of quantum arrays, the three-qubit repetition code, and the nine-qubit Shor code illustrate how quantum data and operations can be organised and specified within the framework.
We also established conditions under which data type quantisation preserves embeddings, isomorphisms, and products.

As a first step toward introducing data abstraction into quantum programming, in this paper, we examine only two typical forms of quantisation, namely bit oracles and phase oracles. One can expect that more forms of quantisation will emerge in the development of future quantum algorithms and programs. We also consider only purely quantum abstract data types in this paper. In practical applications, quantum and classical data are often involved in the same program, and measurement outcomes of quantum data are stored as classical data that subsequently participate in classical computations. This observation suggests another interesting topic for future research: classical-quantum hybrid abstract data types.

We can expect that data abstraction in quantum programming enables programmers to reason about complex quantum behaviour in a manageable way as in classical programming. By hiding mathematical and physical complexity while enforcing quantum constraints, abstraction makes quantum software development more accessible, safer, and more scalable. As quantum computing matures, well-designed abstractions will be crucial for bridging the gap between theoretical algorithms and practical implementations.

\begin{acks}
    This work was supported in part by the Australian Research Council under Grant DP250102952.
\end{acks}

\bibliographystyle{ACM-Reference-Format}
\bibliography{main}

\newpage

\appendix

\section{Deferred Proofs}\label[appendix]{appendix}

All of the theorems and propositions in the main text of this paper were presented without proofs. In this appendix, we give all of these proofs. 
We first prove the following lemma, which is useful for proving \Cref{thm:quantise}.

\begin{lemma}\label{lem:quantised-terms}
Let $A$ and $\calA$ be as in \Cref{thm:quantise}.
Let $t$ be a classical term involving classical variables $x_1:s_1,\ldots,x_n:s_n$ and
$p$ be the principal variable of $\eta(t)$.
We use $\overline{q}=q_1,\ldots,q_n$ to denote the designated variables $q_{x_1},\ldots,q_{x_n}$ as in \Cref{def:quantisation-of-terms}.
Denote $\overline{a}=a_1,\ldots,a_n\in A_{s_1},\ldots,A_{s_n}$.
Define $\sigma_{\overline{a}}$ to be an assignment of classical variables with $\sigma_{\overline{a}}(x_i)=a_i$ for $i=1,\ldots,n$.
For any quantum state 
\begin{equation*}
    \ket{\varphi}_{\overline{q}}=\sum_{\overline{a}}\beta_{\overline{a}} \ket{\overline{a}}_{\overline{q}},
\end{equation*}
we have
\begin{equation*}
    \mathit{eval}_{\ket{\varphi}}(\eta(t))\downarrow p
    =\sum_{\overline{a}}\abs*{\beta_{\overline{a}}}^2 \ket{\mathit{eval}_{\sigma_{\overline{a}}}(t)}\!\bra{\mathit{eval}_{\sigma_{\overline{a}}}(t)}_{p},
\end{equation*}
where $\mathit{eval}$ on the LHS is defined in the quantum algebra $\calA$,
and $\mathit{eval}$ on the RHS is defined in the classical algebra $A$.
\end{lemma}

\begin{proof}
We first consider the lemma for the computational-basis states; that is, $\ket{\varphi}=\ket{\overline{a}}$ for some $\overline{a}$. This case can be proved by induction on the length of $t$ with the following induction hypothesis:
\begin{itemize}
    \item
    \textbf{IH}:
    Let $\overline{o}$ denote the variables involved in the quantisation $\eta(t)$ but not in $p,\overline{q}$.
    For any classical assignment $\sigma$ of $\overline{x}$, let $a_i=\sigma(x_i)$ and $\overline{a}=a_1,\ldots,a_n$.
    There exists a quantum state $\ket{\gamma}$ such that:
    \begin{equation}
        \label{eq:IH}
        \mathit{eval}_{\ket{\overline{a}}}(\eta(t))=\ket{\overline{a}}_{\overline{q}}\ket{\mathit{eval}_{\sigma}(t)}_{p} \ket{\gamma}_{\overline{o}}.
    \end{equation}
\end{itemize}
There are three cases:
\begin{itemize}
    \item
    If $t=c$, then $\eta(t)=\ket{\psi_c}_{p}$ and $\ket{\psi_c}^{\calA}=\ket{c^A}$. 
    Thus, \Cref{eq:IH} holds.
    \item
    If $t=x$, then $\eta(x)=\mathit{SUM}_{s}[q_x,q]\ket{0}_{q}$.
    Since $k_s(0_s^A)=0$ and $a\boxplus_s 0_s^A=a$, we have
    \begin{equation*}
        \mathit{eval}_{\ket{\overline{a}}}(\eta(x))=\mathit{SUM}_{s}[q_x,q]^{\calA}\ket{a}_{q_x}\ket{0_s^A}_{q}=\ket{a}_{q_x}\ket{a}_q.
    \end{equation*}
    Thus, \Cref{eq:IH} holds.
    \item
    If $t=O(t_1,\ldots,t_m)$, then
    \begin{equation*}
        \eta(t)=C_1;\ldots ; C_m;\mathit{pre}_O[p_1,\ldots,p_m,\overline{r}]; U_O[p_1,\ldots,p_m,\overline{r}];\mathit{post}_O[p_1,\ldots,p_m,\overline{r}]\ket{\psi_1} \ldots \ket{\psi_m}\ket{\overline{0}}_{\overline{r}},
    \end{equation*}
    where each $\eta(t_i)=C_i\ket{\psi_i}$ has principal variable $p_i$.
    Note that in the in-place case, $p=p_i$ for the designated index $i$ in \Cref{def:quantisation-of-signature};
    in the out-of-place case, $p$ is the first variable in $\overline{r}$.
    By the \textbf{IH} for $t_1,\ldots,t_m$, there exists a $\ket{\Gamma}$ such that
    \begin{equation}
        \label{eq:eval_C}
        \mathit{eval}_{\ket{\overline{a}}}(C_1;\ldots;C_m\ket{\psi_1}\ldots\ket{\psi_m})=
        \ket{\overline{a}}_{\overline{q}}\ket{b_1,\ldots,b_m}_{p_1,\ldots,p_m}\ket{\Gamma}_{\overline{o_1},\ldots,\overline{o_m}},
    \end{equation}
    where each $b_i=\mathit{eval}_{\sigma}(t_i)$,
    and $\overline{o_i}$ are additional variables involved in $\eta(t_i)$ but not in $p_i,\overline{q}$.
    By \Cref{eq:quantisation-condition}, we further have
    \begin{equation}
        \label{eq:eval_O}
        \mathit{post}_O^{\calA}U_O^{\calA}\mathit{pre}_{O}^{\calA}\ket{b_1,\ldots,b_m}_{p_1,\ldots,p_m}\ket{\overline{0}}_{\overline{r}}
        =\ket{O^A(b_1,\ldots,b_m)}_{p}\ket{\gamma'}
    \end{equation}
    for some $\ket{\gamma'}$.
    Combining \Cref{eq:eval_C,eq:eval_O}, we obtain
    \begin{equation*}
        \mathit{eval}_{\ket{\overline{a}}}(\eta(t))=
        \ket{\overline{a}}_{\overline{q}}\ket{\mathit{eval}_{\sigma}(t)}_{p}(\ket{\Gamma}\ket{\gamma'})_{\overline{o}},
    \end{equation*}
    and \Cref{eq:IH} holds.
\end{itemize}

Since all $\ket{\overline{a}}_{\overline{q}}$ are orthogonal, by linearity and the tracing out of $\overline{q}$, it is easy to prove the lemma for any quantum state $\ket{\varphi}$.
\end{proof}

\begin{proof}[Proof of \Cref{thm:quantise}] 

Since $\calA$ is already a quantum algebra of $\mathit{QSig}$, it suffices to show that $\calA$ satisfies every equation in $\eta(E)$. Consider any equation in $E$:
\begin{equation*}
    e=(\forall x_1:s_1,\ldots,x_n:s_n)(t_1=t_2).
\end{equation*}
We use $\overline{q}=q_1,\ldots,q_n$ to denote the designated variables $q_{x_1},\ldots,q_{x_n}$ corresponding to $x_1,\ldots,x_n$, as in \Cref{def:quantisation-of-terms}.

Let $p$ be the common principal variable of $\eta(t_1)$ and $\eta(t_2)$, as in \Cref{def:quantisation-of-equations}. Consider an arbitrary quantum state
\[
    \ket{\varphi}_{\overline q} = \sum_{\overline a}\beta_{\overline a} \ket{\overline a}_{\overline q}.
\]

Since $A$ is an algebra of $\mathit{Spec}$ and $e\in E$, we have $A\models e$. Therefore, for every $\overline a\in A_{s_1}\times\cdots\times A_{s_n}$,
\begin{equation}
\label{eq:classical-validity}
    \mathit{eval}_{\sigma_{\overline a}}(t_1) = \mathit{eval}_{\sigma_{\overline a}}(t_2).
\end{equation}

Applying \Cref{lem:quantised-terms} to $t_1$ and $t_2$, and using
\Cref{eq:classical-validity}, we obtain
\begin{equation*}
    \mathit{eval}_{\ket{\varphi}}(\eta(t_1)) \downarrow p
    = \mathit{eval}_{\ket{\varphi}}(\eta(t_2)) \downarrow p,
\end{equation*}
which is essentially $\calA\models\eta(e)$ by the validity of equations (see \Cref{def:validity}).
The conclusion then follows.
\end{proof}

\begin{proof}[Proof of \Cref{prop-array}]

Note that in \Cref{eq:qarray-spec}, all variables are principal variables.
To show that the equation is valid in algebra $\calA$, by linearity, 
it suffices to consider any computational-basis state of variables $q_a,q_i,q_v,q_u$:
\begin{equation*}
    \ket{\varphi}=\ket{a}_{q_a}\ket{j}_{q_i}\ket{v}_{q_v}\ket{u}_{q_u}.
\end{equation*}
Let $a'$ be an array with $a'[i]=a[i]$ for $i\neq j$ and $a'[j]=u$.
By evaluating the LHS of \Cref{eq:qarray-spec}, we have
\begin{align*}
    &\mathit{eval}_{\ket{\varphi}}(\mathit{qwrite}[q_a,q_i,q_u];\mathit{qread}[q_a,q_i,q_v])\\
    &\quad =\mathit{qread}[q_a,q_i,q_v]\parens*{\mathit{qwrite}[q_a,q_i,q_u]\ket{a}_{q_a}\ket{j}_{q_i}\ket{v}_{q_v}\ket{u}_{q_u}}\\
    &\quad =\mathit{qread}[q_a,q_i,q_v]\parens*{\ket{a'}_{q_a}\ket{j}_{q_i}\ket{v}_{q_v}\ket{a[j]}_{q_u}}\\
    &\quad =\ket{a'}_{q_a}\ket{j}_{q_i}\ket{v\boxplus u}_{q_v}\ket{a[j]}_{q_u}.
\end{align*}
By evaluating the RHS of \Cref{eq:qarray-spec}, we have
\begin{align*}
    &\mathit{eval}_{\ket{\varphi}}(\mathit{SUM}[q_u,q_v];\mathit{qwrite}[q_a,q_i,q_u])\\
    &\quad =\mathit{qwrite}[q_a,q_i,q_u]\parens*{\mathit{SUM}[q_u,q_v]\ket{a}_{q_a}\ket{j}_{q_i}\ket{v}_{q_v}\ket{u}_{q_u}}\\
    &\quad =\mathit{qwrite}[q_a,q_i,q_u]\parens*{\ket{a}_{q_a}\ket{j}_{q_i}\ket{v\boxplus u}_{q_v}\ket{u}_{q_u}}\\
    &\quad =\ket{a'}_{q_a}\ket{j}_{q_i}\ket{v\boxplus u}_{q_v}\ket{a[j]}_{q_u}.
\end{align*}
By linearity, we have
\begin{equation*}
    \mathit{eval}_{\ket{\varphi}}(\mathit{qwrite}[q_a,q_i,q_u];\mathit{qread}[q_a,q_i,q_v])
    =\mathit{eval}_{\ket{\varphi}}(\mathit{SUM}[q_u,q_v];\mathit{qwrite}[q_a,q_i,q_u])
\end{equation*}
for any quantum state $\ket{\varphi}$, and the conclusion follows.
    
\end{proof}

\begin{proof}[Proof of \Cref{prop:qec-repetition}]
Note that in \Cref{eq:qec-repetition-spec-1,eq:qec-repetition-spec-2}, $q_a$ is the only principal variable. By \Cref{eq:classical-repetition-correction,eq:syn-repetition} and the definition of $\mathit{qcorr}$, for every $x\in\braces*{0,1}$ and $i\in\braces*{0,1,2}$, we have
\begin{align}
    \label{eq:qcorr-no-error}
    \begin{split}
    &\mathit{qcorr}[q_a,q_s]\,
    \ket{\mathit{enc}^{A}(x)}_{q_a}\ket{\bot}_{q_s}\\
    &\quad = \ket{\mathit{corr}^{A}(\mathit{enc}^{A}(x))}_{q_a} \ket{\mathit{syn}(\mathit{enc}^{A}(x))}_{q_s}
    = \ket{\mathit{enc}^{A}(x)}_{q_a}\ket{\bot}_{q_s},
    \end{split}
    \\
    \label{eq:qcorr-one-error}
    \begin{split}
    &\mathit{qcorr}[q_a,q_s]\, \ket{\mathit{flip}^{A}(\mathit{enc}^{A}(x),i)}_{q_a} \ket{\bot}_{q_s}\\
    &\quad = \ket{\mathit{corr}^{A} (\mathit{flip}^{A}(\mathit{enc}^{A}(x),i))}_{q_a} \ket{\mathit{syn} (\mathit{flip}^{A}(\mathit{enc}^{A}(x),i))}_{q_s} = \ket{\mathit{enc}^{A}(x)}_{q_a}\ket{i}_{q_s}.
    \end{split}
\end{align}

To prove that \Cref{eq:qec-repetition-spec-1} is valid in $\calA$, consider an arbitrary state of the variable $q_l$:
\begin{equation*}
    \ket{\varphi}_{q_l}
    =\sum_{x\in\braces*{0,1}}\alpha_x\ket{x}_{q_l}.
\end{equation*}
By evaluating the LHS and using \Cref{eq:qcorr-no-error}, we have
\begin{equation*}
    \mathit{eval}_{\ket{\varphi}} \parens*{ \mathit{qenc}[q_l,q_a];\mathit{qcorr}[q_a,q_s] \ket{0}_{q_a}\ket{\bot}_{q_s}} = \ket{0}_{q_l} \sum_x\alpha_x\ket{x,x,x}_{q_a}\ket{\bot}_{q_s}.
\end{equation*}
By evaluating the RHS, we have
\begin{equation*}
    \mathit{eval}_{\ket{\varphi}} \parens*{ \mathit{qenc}[q_l,q_a] \ket{0}_{q_a}\ket{\bot}_{q_s}} = \ket{0}_{q_l} \sum_x\alpha_x\ket{x,x,x}_{q_a}\ket{\bot}_{q_s}.
\end{equation*}
Both evaluations give the same state, so \Cref{eq:qec-repetition-spec-1} is valid.

To prove that \Cref{eq:qec-repetition-spec-2} is valid in $\calA$, consider an arbitrary state of variables $q_l,q_i$:
\begin{equation*}
    \ket{\varphi}_{q_lq_i} = \sum_{\substack{x\in\braces*{0,1}\\
    i\in\braces*{0,1,2}}} \alpha_{x,i}\ket{x}_{q_l}\ket{i}_{q_i}.
\end{equation*}
By evaluating the LHS and using \Cref{eq:qcorr-one-error}, we have
\begin{equation*}
    \mathit{eval}_{\ket{\varphi}} \parens*{ \mathit{qenc}[q_l,q_a]; \mathit{qflip}[q_a,q_i]; \mathit{qcorr}[q_a,q_s] \ket{0}_{q_a}\ket{\bot}_{q_s}} = \ket{0}_{q_l} \sum_{x,i}\alpha_{x,i} \ket{x,x,x}_{q_a}\ket{i}_{q_i}\ket{i}_{q_s}.
\end{equation*}
By evaluating the RHS, we have
\begin{equation*}
    \mathit{eval}_{\ket{\varphi}} \parens*{ \mathit{qenc}[q_l,q_a] \ket{0}_{q_a}\ket{\bot}_{q_s}} = \ket{0}_{q_l} \sum_{x,i}\alpha_{x,i} \ket{x,x,x}_{q_a}\ket{i}_{q_i}\ket{\bot}_{q_s}.
\end{equation*}
Since $q_a$ is the principal variable, we trace out $q_l,q_i,q_s$.
Both sides of \Cref{eq:qec-repetition-spec-2} give
\begin{equation*}
    \sum_{i=0}^{2}
    \sum_{x,x'\in\braces*{0,1}}
    \alpha_{x,i}\overline{\alpha_{x',i}}
    \ket{x,x,x}\!\bra{x',x',x'}_{q_a}.
\end{equation*}
Therefore, \Cref{eq:qec-repetition-spec-2} is valid.
The conclusion immediately follows.
\end{proof}

\begin{lemma}
    \label{lem:exist-Phi}
    There exists an orthonormal basis $\braces*{\ket{\Phi_{x,s}}:x\in \braces{0,1},s\in \braces{0,1,2,\bot}^4}$ that satisfies \Cref{eq:def-Phi}.
\end{lemma}

\begin{proof}
    We first consider the following orthonormal basis of $\parens*{\calH_{\texttt{pqbit}}^{\calB}}^{\otimes 3}$:
    \begin{equation*}
        \braces*{
        \frac{\ket{000}\pm\ket{111}}{\sqrt{2}},
        \frac{\ket{100}\pm\ket{011}}{\sqrt{2}},
        \frac{\ket{010}\pm\ket{101}}{\sqrt{2}},
        \frac{\ket{001}\pm\ket{110}}{\sqrt{2}}
        }
    \end{equation*}
    Let $i=3r+t$, where $r,t\in\braces*{0,1,2}$. For
    $P\in\braces*{\Id,X,Y,Z}$, applying $P_i$ to the state in \Cref{eq:shor-code-word} gives
    \begin{equation}
        \label{eq:psi_state}
        P_i\ket{\Psi_x} =
        \parens*{\frac{\ket{000}+(-1)^x\ket{111}}{\sqrt{2}}}^{\otimes r} \otimes
        P_t\parens*{\frac{\ket{000}+(-1)^x\ket{111}}{\sqrt{2}}} \otimes
        \parens*{\frac{\ket{000}+(-1)^x\ket{111}}{\sqrt{2}}}^{\otimes(2-r)},
    \end{equation}
    where $P_t$ applies $P$ on the qubit indexed by $t$ within a group of three qubits.
    In particular, for different $P_t$, we have
    \begin{equation}
        \label{eq:three-action-P_t}
        \begin{split}
            X_t\parens*{\frac{\ket{000}+(-1)^x\ket{111}}{\sqrt{2}}}
            &=
            \begin{cases}
                \dfrac{\ket{100}+(-1)^x\ket{011}}{\sqrt{2}}, &t=0,\\
                \dfrac{\ket{010}+(-1)^x\ket{101}}{\sqrt{2}}, &t=1,\\
                \dfrac{\ket{001}+(-1)^x\ket{110}}{\sqrt{2}}, &t=2,
            \end{cases}\\
            Z_t\parens*{\frac{\ket{000}+(-1)^x\ket{111}}{\sqrt{2}}}
            &=\frac{\ket{000}-(-1)^x\ket{111}}{\sqrt{2}},\\
            Y_t\parens*{\frac{\ket{000}+(-1)^x\ket{111}}{\sqrt{2}}}
            &=\mathrm{i}
            \begin{cases}
                \dfrac{\ket{100}-(-1)^x\ket{011}}{\sqrt{2}}, &t=0,\\
                \dfrac{\ket{010}-(-1)^x\ket{101}}{\sqrt{2}}, &t=1,\\
                \dfrac{\ket{001}-(-1)^x\ket{110}}{\sqrt{2}}, &t=2.
            \end{cases}
        \end{split}
    \end{equation}

    From \Cref{eq:sigma}, $\sigma(\Id,i)=(\bot,\bot,\bot,\bot)$ for every $i\in B_{\texttt{idx}}$, and $\Id_i\ket{\Psi_x}=\ket{\Psi_x}$.
    Apart from the case $P=P'=\Id$, $\sigma(P,i)=\sigma(P',i')$ for distinct pairs $(P,i)\neq(P',i')$ only if $P=P'=Z$ and $i,i'\in\braces*{3r,3r+1,3r+2}$ for some $r\in\braces*{0,1,2}$. 
    Moreover, for every $x\in B_{\texttt{lbit}}$ and $r\in\braces*{0,1,2}$, we have
    \begin{equation*}
        Z_{3r}\ket{\Psi_x} =Z_{3r+1}\ket{\Psi_x} =Z_{3r+2}\ket{\Psi_x}.
    \end{equation*}
    Thus, for every $x\in B_{\texttt{lbit}}$, $P,P'\in\braces*{\Id,X,Y,Z}$, and $i,i'\in B_{\texttt{idx}}$,
    \begin{equation*}
        \sigma(P,i)=\sigma(P',i') \Rightarrow P_i\ket{\Psi_x}=P'_{i'}\ket{\Psi_x}.
    \end{equation*}

    By \Cref{eq:sigma,eq:psi_state,eq:three-action-P_t}, for every $x,x'\in B_{\texttt{lbit}}$, $P,P'\in\braces*{\Id,X,Y,Z}$, and $i,i'\in B_{\texttt{idx}}$, it is easy to verify that
    \begin{equation*}
        \bra{\Psi_x}P_i^{\dagger}P'_{i'}\ket{\Psi_{x'}}
        =
        \begin{cases}
            1, &x=x'\text{ and }\sigma(P,i)=\sigma(P',i'),\\
            0, &o.w.
        \end{cases}
    \end{equation*}
    Therefore, \Cref{eq:def-Phi} is well-defined and these $\ket{\Phi_{x,\sigma(P,i)}}$ can be extended to an orthonormal basis of $\calH_{\texttt{qarray}}^{\calB}$.
\end{proof}

\begin{proof}[Proof of \Cref{prop:qec-shor}]
Note that in \Cref{eq:qec-shor-spec-1,eq:qec-shor-spec-2,eq:qec-shor-spec-3,eq:qec-shor-spec-4}, $q_a$ is the only principal variable. 
By \Cref{eq:sigma,eq:def-Phi} and the definition of $\mathit{qcorr}$,
for every $x\in B_{\texttt{lbit}}$, we have
\begin{equation}
    \label{eq:qcorr-shor-no-error}
    \begin{split}
    &\mathit{qcorr}[q_a,q_s]\,
    \ket{\Psi_x}_{q_a}\ket{\bot,\bot,\bot,\bot}_{q_s}\\
    &\quad =\mathit{qcorr}[q_a,q_s]\,
    \ket{\Phi_{x,(\bot,\bot,\bot,\bot)}}_{q_a}
    \ket{\bot,\bot,\bot,\bot}_{q_s}
    =\ket{\Psi_x}_{q_a}\ket{\bot,\bot,\bot,\bot}_{q_s}.
    \end{split}
\end{equation}
Similarly, for every $x\in B_{\texttt{lbit}}$, $P\in\braces*{X,Y,Z}$, and $i\in B_{\texttt{idx}}$,
\begin{equation}
    \label{eq:qcorr-shor-one-error}
    \begin{split}
    &\mathit{qcorr}[q_a,q_s]\,
    P_i\ket{\Psi_x}_{q_a}\ket{\bot,\bot,\bot,\bot}_{q_s}\\
    &\quad =\mathit{qcorr}[q_a,q_s]\,
    \ket{\Phi_{x,\sigma(P,i)}}_{q_a}
    \ket{\bot,\bot,\bot,\bot}_{q_s}
    =\ket{\Psi_x}_{q_a}\ket{\sigma(P,i)}_{q_s}.
    \end{split}
\end{equation}

To prove that \Cref{eq:qec-shor-spec-1} is valid in $\calB$, consider an arbitrary state of the variable $q_l$:
\begin{equation*}
    \ket{\varphi}_{q_l}
    =\sum_{x\in\braces*{0,1}}\alpha_x\ket{x}_{q_l}.
\end{equation*}
By evaluating the LHS and using \Cref{eq:qcorr-shor-no-error}, we have
\begin{equation*}
    \mathit{eval}_{\ket{\varphi}}
    \parens*{\mathit{qenc}[q_l,q_a];\mathit{qcorr}[q_a,q_s]
    \ket{0}_{q_a}\ket{\bot,\bot,\bot,\bot}_{q_s}}
    =
    \ket{0}_{q_l}\sum_x\alpha_x\ket{\Psi_x}_{q_a}
    \ket{\bot,\bot,\bot,\bot}_{q_s}.
\end{equation*}
By evaluating the RHS, we have
\begin{equation*}
    \mathit{eval}_{\ket{\varphi}}
    \parens*{\mathit{qenc}[q_l,q_a]
    \ket{0}_{q_a}\ket{\bot,\bot,\bot,\bot}_{q_s}}
    =
    \ket{0}_{q_l}\sum_x\alpha_x\ket{\Psi_x}_{q_a}
    \ket{\bot,\bot,\bot,\bot}_{q_s}.
\end{equation*}
Both evaluations give the same state, so
\Cref{eq:qec-shor-spec-1} is valid.

To prove that
\Cref{eq:qec-shor-spec-2,eq:qec-shor-spec-3,eq:qec-shor-spec-4} are valid in $\calB$,
consider an arbitrary state of variables $q_l,q_i$:
\begin{equation*}
    \ket{\varphi}_{q_lq_i}
    =
    \sum_{\substack{x\in\braces*{0,1}\\
                    i\in\braces*{0,1,\ldots,8}}}
    \alpha_{x,i}\ket{x}_{q_l}\ket{i}_{q_i}.
\end{equation*}
By evaluating the LHSs and using \Cref{eq:qcorr-shor-one-error}, we have
\begin{align*}
    &\mathit{eval}_{\ket{\varphi}}
    \parens*{\mathit{qenc}[q_l,q_a];\mathit{qflipX}[q_a,q_i];
    \mathit{qcorr}[q_a,q_s]
    \ket{0}_{q_a}\ket{\bot,\bot,\bot,\bot}_{q_s}}\\
    &\quad =\ket{0}_{q_l}\sum_{x,i}\alpha_{x,i}
    \ket{\Psi_x}_{q_a}\ket{i}_{q_i}\ket{\sigma(X,i)}_{q_s},\\
    &\mathit{eval}_{\ket{\varphi}}
    \parens*{\mathit{qenc}[q_l,q_a];\mathit{qflipZ}[q_a,q_i];
    \mathit{qcorr}[q_a,q_s]
    \ket{0}_{q_a}\ket{\bot,\bot,\bot,\bot}_{q_s}}\\
    &\quad =\ket{0}_{q_l}\sum_{x,i}\alpha_{x,i}
    \ket{\Psi_x}_{q_a}\ket{i}_{q_i}\ket{\sigma(Z,i)}_{q_s},\\
    &\mathit{eval}_{\ket{\varphi}}
    \parens*{\mathit{qenc}[q_l,q_a];\mathit{qflipX}[q_a,q_i];
    \mathit{qflipZ}[q_a,q_i];\mathit{qcorr}[q_a,q_s]
    \ket{0}_{q_a}\ket{\bot,\bot,\bot,\bot}_{q_s}}\\
    &\quad =\mathrm{i}\ket{0}_{q_l}\sum_{x,i}\alpha_{x,i}
    \ket{\Psi_x}_{q_a}\ket{i}_{q_i}
    \ket{\sigma(Y,i)}_{q_s}.
\end{align*}
By evaluating the RHS, we have
\begin{equation*}
    \mathit{eval}_{\ket{\varphi}}
    \parens*{\mathit{qenc}[q_l,q_a]
    \ket{0}_{q_a}\ket{\bot,\bot,\bot,\bot}_{q_s}}
    =
    \ket{0}_{q_l}\sum_{x,i}\alpha_{x,i}
    \ket{\Psi_x}_{q_a}\ket{i}_{q_i}
    \ket{\bot,\bot,\bot,\bot}_{q_s}.
\end{equation*}
Since $q_a$ is the principal variable, we trace out $q_l,q_i,q_s$.
Both sides of \Cref{eq:qec-shor-spec-2,eq:qec-shor-spec-3,eq:qec-shor-spec-4} give
\begin{equation*}
    \sum_{i=0}^{8}
    \sum_{x,x'\in\braces*{0,1}}
    \alpha_{x,i}\overline{\alpha_{x',i}}
    \ket{\Psi_x}\!\bra{\Psi_{x'}}_{q_a}.
\end{equation*}
Therefore, \Cref{eq:qec-shor-spec-2,eq:qec-shor-spec-3,eq:qec-shor-spec-4} are valid.
The conclusion immediately follows.
\end{proof}

\begin{proof}[Proof of \Cref{prop:BO-embedding}]
As observed after \Cref{eq-emb}, every $G_{s_q}$ is an isometric embedding.
It suffices to verify the two conditions in \Cref{def-qemb}.
\begin{enumerate}
\item
Let $\ket{\psi}\in\calK$ be a quantum state symbol.
Then, there is a constant symbol $c:\rightarrow s$ such that $\ket{\psi}=\eta(c)$. 
Since $g$ is a homomorphism, $g_s(c^A)=c^{A'}$, which means
\[
G_{s_q}\ket{\psi}^{\calA}
 =G_{s_q}\ket{c^A}
 =\ket{c^{A'}}
 =\ket{\psi}^{\calA'}.
\]
Consequently, \Cref{eq:emb-cond-1} holds.

\item
Let $U\in\calU$ be a quantum gate symbol. It suffices to consider the following cases.

\begin{itemize}

\item
Suppose $U=\mathit{pre}_O$ or $U=\mathit{post}_O$ for some
nonconstant operation symbol $O$. 
This case is trivial because $\mathit{pre}_O^{\calB}=\mathit{post}_O^{\calB}=\Id$ for $\calB\in \braces{\calA,\calA'}$.

\item
Suppose $U=U_O$ for some nonconstant operation symbol $O:s_1,\ldots,s_n\rightarrow s$. Write $s_{iq}=\eta(s_i)$ and $s_q=\eta(s)$. For
$a_i\in A_{s_i}$ and $a\in A_s$, we have
\begin{align*}
&\parens*{ G_{s_{1q}}\otimes\ldots\otimes G_{s_{nq}}\otimes G_{s_q} } U_O^{\calA} \ket{a_1}\cdots\ket{a_n}\ket{a}\\
&\quad= \ket{g_{s_1}(a_1)}\cdots\ket{g_{s_n}(a_n)} \ket{g_s\parens{a\boxplus_s^A O^A(a_1,\ldots,a_n)}}\\
&\quad= \ket{g_{s_1}(a_1)}\cdots\ket{g_{s_n}(a_n)} \ket{g_s(a)\boxplus_s^{A'} O^{A'}(g_{s_1}(a_1),\ldots,g_{s_n}(a_n))}\\
&\quad= U_O^{\calA'} \parens*{ G_{s_{1q}}\otimes\ldots\otimes G_{s_{nq}}\otimes G_{s_q} } \ket{a_1}\cdots\ket{a_n}\ket{a}.
\end{align*}
The second equality follows from the condition about $\boxplus_s^A$ and $\boxplus_s^{A'}$ and the fact that $g$ is a homomorphism.
\end{itemize}
By linearity, \Cref{eq-hom} holds.
\end{enumerate}

Thus, $G$ satisfies \Cref{def-qemb} and is an embedding from
$\calA$ to $\calA'$.
\end{proof}

\begin{proof}[Proof of \Cref{prop:PO-embedding}]
As observed after \Cref{eq-emb}, every $G_{s_q}$ is unitary. It suffices to verify the two conditions in
\Cref{def-qemb}. For each $s\in S$, let $N_s=\abs*{A_s}=\abs*{A_s'}$ and $\omega_s=e^{2\pi i/N_s}$.
\begin{enumerate}
\item
Let $\ket{\psi}\in \calK$ be a quantum state symbol.
Similar to the proof of \Cref{prop:BO-embedding},
$G_{s_q}\ket{\psi}^{\calA}=\ket{\psi}^{\calA'}$ and therefore \Cref{eq:emb-cond-1} holds.

\item
Let $U\in\calU$ be a quantum gate symbol. It suffices to consider the
following cases.
\begin{itemize}

\item
Suppose $U=\mathit{pre}_O$ or $\mathit{post}_O$ for some nonconstant operation symbol $O:s_1,\ldots,s_n\rightarrow s$. For any $a\in A_s$, since $g_s$ is bijective,
\begin{equation}
    \label{eq:PO-QFT}
    \begin{split}
        G_{s_q}\mathit{QFT}_{\abs{A_s}}\ket{a}
        &=\frac{1}{\sqrt{\abs{A_s}}}\sum_{b\in A_s}
        \omega_s^{k_s^A(a)k_s^A(b)}\ket{g_s(b)}\\
        &=\frac{1}{\sqrt{\abs{A_s}}}\sum_{g_s(b)\in A_s'}
        \omega_s^{k_s^{A'}(g_s(a))k_s^{A'}(g_s(b))}\ket{g_s(b)}\\
        &=\mathit{QFT}_{\abs{A_s}} G_{s_q}\ket{a}.
    \end{split}
\end{equation}
Here, the second equality uses $k_s^A=k_s^{A'}\circ g_s$.
Tensoring \Cref{eq:PO-QFT} with $\bigotimes_{i=1}^nG_{s_{iq}}$ followed by linearity leads to \Cref{eq-hom} for $U=\mathit{pre}_O$. Taking the adjoints proves it for $U=\mathit{post}_O$.

\item
Suppose $U=U_O$ for some nonconstant operation symbol $O:s_1,\ldots,s_n\rightarrow s$. 
For $a_i\in A_{s_i}$ and $a\in A_s$, \Cref{eq:po-oracle} gives
\begin{align*}
&(G_{s_{1q}}\otimes \ldots \otimes G_{s_{nq}}\otimes G_{s_q}) U_O^{\calA}\ket{a_1,\ldots,a_n}\ket{a}\\
&=\omega_s^{k_s^A(a)k_s^A(O^A(a_1,\ldots,a_n))} \ket{g_{s_1}(a_1),\ldots,g_{s_n}(a_n)}\ket{g_s(a)}\\
&=\omega_s^{ k_s^{A'}(g_s(a))k_s^{A'}(O^{A'}(g_{s_1}(a_1),\ldots,g_{s_n}(a_n)))} \ket{g_{s_1}(a_1),\ldots,g_{s_n}(a_n)}\ket{g_s(a)}\\
&=U_O^{\calA'}(G_{s_{1q}}\otimes \ldots\otimes G_{s_{nq}}\otimes G_{s_q}) \ket{a_1,\ldots,a_n}\ket{a}.
\end{align*}
The second equality follows from $k_s^{A'}\circ g_s=k_s^A$ and the fact that $g$ is a homomorphism.
By linearity, \Cref{eq-hom} holds.
\end{itemize}
\end{enumerate}

Thus, $G$ satisfies \Cref{def-qemb} and is an isomorphism from $\calA$ to $\calA'$.
\end{proof}

\begin{proof}[Proof of \Cref{prop:BO-product}]
As observed after \Cref{eq:product-isomorphism}, every $J_{s_q}$ is unitary. It suffices to verify the two conditions in \Cref{def-qemb}.
\begin{enumerate}
\item
Let $\ket{\psi}\in\calK$ be a quantum state symbol, so $\ket{\psi}=\eta(c)$ for a constant symbol $c:\rightarrow s$. Since $\ket{\psi}$ has the single sort $s_q$, $T_{s_q}=\Id$. Then
    \begin{equation*}
        J_{s_q}\ket{\psi}^{\calC} =J_{s_q}\ket{(c^A,c^B)} =\ket{c^A}\otimes\ket{c^B} =\ket{\psi}^{\calA\otimes\calB}.
    \end{equation*}
Consequently, \Cref{eq:emb-cond-1} holds.

\item
Let $U\in\calU$ be a quantum gate symbol. It suffices to consider the following cases.
\begin{itemize}

\item
Suppose $U=\mathit{pre}_O$ or $U=\mathit{post}_O$ for some nonconstant operation symbol $O$. This case is trivial because $\mathit{pre}_{O}^{\calD}=\mathit{post}_{O}^{\calD}=\Id$ for $\calD\in \braces{\calA,\calB,\calC}$, and \Cref{u-p} further implies that $\mathit{pre}_{O}^{\calA\otimes\calB}=\mathit{post}_{O}^{\calA\otimes\calB}=\Id$. Therefore, \Cref{eq-hom} holds in both cases.

\item
Suppose $U=U_O$ for some nonconstant operation symbol $O:s_1,\ldots,s_n\rightarrow s$. 
Let us denote $\overline{a}=a_1,\ldots,a_n$, $\overline{b}=b_1,\ldots,b_n$, and $\overline{(a,b)}=(a_1,b_1),\ldots,(a_n,b_n)$.  For $a_i\in A_{s_i}$, $b_i\in B_{s_i}$, $a\in A_s$, and $b\in B_s$, we have 
\begin{align*}
&(J_{s_{1q}}\otimes \ldots\otimes J_{s_{nq}}\otimes J_{s_q})U_O^{\calC} \ket{\overline{(a,b)}}\ket{(a,b)}\\
&=(J_{s_{1q}}\otimes \ldots\otimes J_{s_{nq}}\otimes J_{s_q})\ket{\overline{(a,b)}} \ket{(a,b)\boxplus_s^C O^C(\overline{(a,b)})}\\
&=\bigotimes_{i=1}^n(\ket{a_i}\otimes\ket{b_i}) \otimes\parens*{ \ket{a\boxplus_s^A O^A(\overline{a})} \otimes\ket{b\boxplus_s^B O^B(\overline{b})}}\\
&=T_{s_{1q},\ldots,s_{nq},s_q}^{-1} \parens*{U_O^{\calA}\ket{\overline{a}}\ket{a} \otimes U_O^{\calB}\ket{\overline{b}}\ket{b}}\\
&=U_O^{\calA\otimes\calB}(J_{s_{1q}}\otimes \ldots\otimes J_{s_{nq}}\otimes J_{s_q}) \ket{\overline{(a,b)}}\ket{(a,b)}.
\end{align*}
Here, the first equality follows from \Cref{eq:bo-oracle}; the second follows from \Cref{eq:product-isomorphism,c-prod,eq:product-sum-compatible}; the third follows from \Cref{eq:bo-oracle,U-product}; and the last follows from \Cref{u-p,U-product}.
\end{itemize}
By linearity, \Cref{eq-hom} holds.
\end{enumerate}

Thus, $J$ satisfies \Cref{def-qemb} and is an isomorphism from $\calC$ to $\calA\otimes\calB$.
\end{proof}

\begin{proof}[Proof of \Cref{prop:PO-product}]
As shown in the proof of \Cref{prop:BO-product}, \Cref{eq:emb-cond-1} holds. It remains to verify \Cref{eq-hom} for every $U\in\calU$ through the following cases.
\begin{itemize}

\item
Suppose that $U=\mathit{pre}_O$ or $U=\mathit{post}_O$ for some nonconstant operation symbol $O:s_1,\ldots,s_n\rightarrow s$.
\Cref{eq:product-phase-compatible} implies that for each sort $s\in S$,
\begin{equation}
    \label{eq:qft-equality}
    J_{s_q}\mathit{QFT}_{|C_s|} =
    \parens*{ \mathit{QFT}_{|A_s|} \otimes \mathit{QFT}_{|B_s|}} J_{s_q}.
\end{equation}
Consequently,
\begin{align*}
    &(J_{s_{1q}}\otimes \ldots\otimes J_{s_{nq}}\otimes J_{s_q}) \mathit{pre}_O^{\calC}\\
    &=(J_{s_{1q}}\otimes \ldots\otimes J_{s_{nq}}\otimes J_{s_q}) (\Id_{s_{1q},\ldots,s_{nq}}\otimes \mathit{QFT}_{|C_s|})\\
    &=(J_{s_{1q}}\otimes \ldots\otimes J_{s_{nq}})\otimes \parens*{ (\mathit{QFT}_{|A_s|}\otimes \mathit{QFT}_{|B_s|})J_{s_q}}\\
    &=T_{s_{1q},\ldots,s_{nq},s_q}^{-1} (\mathit{pre}_O^{\calA}\otimes \mathit{pre}_O^{\calB}) T_{s_{1q},\ldots,s_{nq},s_q}(J_{s_{1q}}\otimes \ldots\otimes J_{s_{nq}}\otimes J_{s_q})\\
    &=\mathit{pre}_O^{\calA\otimes\calB} (J_{s_{1q}}\otimes \ldots\otimes J_{s_{nq}}\otimes J_{s_q}),
\end{align*}
where the second equality uses \Cref{eq:qft-equality},
the third follows from \Cref{U-product}, and the final follows from \Cref{u-p}.
Combining the above together, \Cref{eq-hom} holds for $U=\mathit{pre}_O$.
Taking adjoints proves it for $U=\mathit{post}_O$.

\item
Suppose $U=U_O$ for some nonconstant operation symbol $O:s_1,\ldots,s_n\rightarrow s$. Write $\overline{a}=a_1,\ldots,a_n$, $\overline{b}=b_1,\ldots,b_n$, and $\overline{(a,b)}=(a_1,b_1),\ldots,(a_n,b_n)$. For $a_i\in A_{s_i}$, $b_i\in B_{s_i}$, $a\in A_s$, and $b\in B_s$, let $a'=O^A(\overline{a})$ and $b'=O^B(\overline{b})$. 
By \Cref{c-prod}, $O^C(\overline{(a,b)})=(a',b')$.
Then we have
\begin{align*}
    &(J_{s_{1q}}\otimes \ldots\otimes J_{s_{nq}}\otimes J_{s_q})
U_O^{\calC}\ket{\overline{(a,b)}}\ket{(a,b)}\\
&=\exp\parens*{\frac{2\pi i\,k_s^C((a,b))k_s^C((a',b'))}{|C_s|}}
(J_{s_{1q}}\otimes \ldots\otimes J_{s_{nq}}\otimes J_{s_q})\ket{\overline{(a,b)}}\ket{(a,b)}\\
&=\exp\parens*{\frac{2\pi i\,k_s^A(a)k_s^A(a')}{|A_s|}}
\exp\parens*{\frac{2\pi i\,k_s^B(b)k_s^B(b')}{|B_s|}}
(J_{s_{1q}}\otimes \ldots\otimes J_{s_{nq}}\otimes J_{s_q})
\ket{\overline{(a,b)}}\ket{(a,b)}\\
&=U_O^{\calA\otimes\calB}(J_{s_{1q}}\otimes \ldots\otimes J_{s_{nq}}\otimes J_{s_q})\ket{\overline{(a,b)}}\ket{(a,b)}.
\end{align*}
Here, the first equality follows from \Cref{eq:po-oracle}; the second follows from \Cref{eq:product-phase-compatible}; and the last follows from \Cref{eq:po-oracle,u-p}. By linearity, \Cref{eq-hom} holds.
\end{itemize}

Thus, $J$ satisfies \Cref{def-qemb} and is an isomorphism from $\calC$ to $\calA\otimes\calB$.
\end{proof}

\end{document}